\documentclass[11pt]{article}
\pdfoutput=1
\usepackage{fullpage}
\usepackage{hyperref}
\usepackage{float}
\usepackage{epsfig}
\usepackage{amsmath}
\usepackage{amssymb}
\usepackage{amstext}
\usepackage{amsmath}
\usepackage{xspace}
\usepackage{latexsym}
\usepackage{verbatim}
\usepackage{multirow}
\usepackage{ifthen}
\usepackage{url}

\usepackage[ruled,vlined]{algorithm2e}

\usepackage{times}
\newboolean{FULL_VERSION}
\setboolean{FULL_VERSION}{true}

\newtheorem{theorem}{Theorem}[section]
\newtheorem{definition}{Definition}

\newtheorem{lemma}[theorem]{Lemma}
\newtheorem{fact}[theorem]{Fact}

\newtheorem{corollary}[theorem]{Corollary}
\newtheorem{example}{Example}
\newtheorem{remark}{Remark}

\newcommand{\shortcite}{\cite}
\newcommand{\qed}{\mbox{\ \ \ }\rule{6pt}{7pt} \bigskip}

\renewcommand{\comment}[1]{}
\newenvironment{proof}{\noindent{\em Proof:}}{\hfill\qed}

\newcommand{\hidecite}[1]{}

\newcommand{\argmax}{\operatorname{argmax}}
\newcommand{\argmin}{\operatorname{argmin}}

\DeclareMathOperator*{\limessinf}{lim\,ess\,inf}
\DeclareMathOperator*{\limesssup}{lim\,ess\,sup}
\DeclareMathOperator*{\esssup}{ess\,sup}
\DeclareMathOperator*{\essinf}{ess\,inf}

\newcommand{\macroPr}[1]{\Pr\left(#1\right)}
\newcommand{\macroCPr}[2]{\macroPr{\left.#1\right|#2}}

\newcommand{\val}{v}

\newcommand{\vali}[1][i]{{\val_{#1}}}
\newcommand{\valj}[1][j]{{\val_{#1}}}
\newcommand{\Val}{V}

\newcommand{\Vali}[1][i]{{\Val_{#1}}}

\newcommand{\price}{P}

\newcommand{\pricei}[1][i]{{\price_{#1}}}

\newcommand{\psc}{\mathcal P}
\newcommand{\psci}[1][i]{\mathcal P_{#1}}
\newcommand{\pibar}{\bar \pi}
\newcommand{\optMidrRedn}{\text{{\tt MIDRtoMech}}}
\newcommand{\BKSRedn}{\text{{\tt SPtoMechBKS}}}

\newcommand{\expect}{\mathbf{E}}

\newcommand{\asc}{\mathcal A}
\newcommand{\asci}{\mathcal A_i}

\newcommand{\alphap}{\alpha_P}
\newcommand{\alphar}{\alpha_R}
\newcommand{\alphaw}{\alpha_W}

\newcommand{\real}{\mathbf R}

\newcommand{\fbks}{{BKS}}
\newcommand{\mub}[1][b]{\mu_{#1}}
\newcommand{\mubpsci}{\nu_{b,i}}
\newcommand{\rhob}[1][\mu]{{\rho_{b}^{#1}}}
\newcommand{\rhobi}[1][\mu]{{\rho_{b,i}^{#1}}}

\renewcommand{\Re}{{\mathbb R}}
\newcommand{\pr}{{\operatorname{Pr}}}
\newcommand{\var}{{\operatorname{Var}}}

\newcommand{\me}{\mathcal{M}}
\newcommand{\al}{A}
\newcommand{\pfn}{f_i}
\newcommand{\ocm}{o}
\newcommand{\Ocm}{\mathcal{O}}

\newcommand{\tyi}{t_i}

\newcommand{\Ty}{T}
\newcommand{\Tyi}{T_i}

\newcommand{\Dist}{\mathcal{D}}
\newcommand{\dist}{D}

\newcommand{\ex}{\mathbf{E}}

\newcommand{\TCSeas}{}

\newcommand{\TCSdd}{}
\newcommand{\TCSe}{}
\newcommand{\TCSeafeed}{}

\newcommand{\rect}{\mathbb{S}}

\begin{document}
\title{Single-Call Mechanisms}

\author{Christopher A. Wilkens\thanks{Computer Science Dept., University of
California at Berkeley. {\tt cwilkens@cs.berkeley.edu}. Supported in part by NSF grant CC-0964033 and by a Google University Research Award. Part of this work was done while the author was an intern at Microsoft Research, Redmond, WA.}
\and Balasubramanian Sivan\thanks{Computer Sciences Dept., University
 of Wisconsin - Madison.  {\tt balu2901@cs.wisc.edu}. Supported in part by NSF award 
  CCF-0830494. Part of this work was done while the author was an intern at Microsoft Research, Redmond, WA.}
}
\date{}
\maketitle{}

\thispagestyle{empty}

\begin{abstract}
Truthfulness is fragile and demanding. It is oftentimes computationally harder than solving the original problem. Even worse, truthfulness can be utterly destroyed by small uncertainties in a mechanism's outcome. One obstacle is that {\em truthful payments depend on outcomes other than the one realized}, such as the lengths of non-shortest-paths in a shortest-path auction. Single-call mechanisms are a powerful tool that circumvents this obstacle --- they implicitly charge truthful payments, guaranteeing truthfulness in expectation using only the outcome realized by the mechanism. The cost of such truthfulness is a trade-off between the expected quality of the outcome and the risk of large payments.

We largely settle when and to what extent single-call mechanisms are possible. The first single-call construction was discovered by Babaioff, Kleinberg, and Slivkins~\cite{BKS10} in single-parameter domains. They give a transformation that turns any monotone, single-parameter allocation rule into a truthful-in-expectation single-call mechanism. Our first result is a natural complement to~\cite{BKS10}: we give a new transformation that produces a single-call VCG mechanism from any allocation rule for which VCG payments are truthful. Second, in both the single-parameter and VCG settings, we precisely characterize the possible transformations, showing that that a wide variety of transformations are possible but that all take a very simple form. Finally, we study the inherent trade-off between the expected quality of the outcome and the risk of large payments. We show that our construction and that of~\cite{BKS10} simultaneously optimize a variety of metrics in their respective domains.

Our study is motivated by settings where uncertainty in a mechanism renders other known techniques untruthful. As an example, we analyze pay-per-click advertising auctions, where the truthfulness of the standard VCG-based auction is easily broken when the auctioneer's estimated click-through-rates are imprecise.

\end{abstract}

\section{Introduction}
\label{sec:intro}

In their seminal work that sparked the field of Algorithmic Mechanism Design,
Nisan and Ronen~\cite{NR01} made a striking observation: na\"{\i}vely
computing VCG payments for shortest-path auctions requires computing ``$n$
versions of the original problem.'' In their case, it requires solving $n+1$
different shortest path problems in a network. Over the next decade, as
researchers studied computation in mechanisms, they repeatedly noticed that
computing payments is harder than solving the original problem. Babaioff et
al.~\cite{BBNS08} exhibited a problem for which deterministic truthfulness
is precisely $(n+1)$-times harder than the original problem. In the case of
Nisan and Ronen's own path auction, Hershberger et al.~\cite{HSB07} showed
that computing VCG prices for a directed graph requires time equivalent to
$\sqrt n$ shortest path computations.\footnote{ Interestingly, the undirected
case is easier. Hershberger and Suri~\cite{HS01,HS02} show that it only
requires time equivalent to a single shortest-path computation. Their work is
orthogonal to our own --- single-call mechanisms achieve truthfulness in a
limited-information setting using only one shortest-path computation,
while~\cite{HSB07,HS01,HS02} assume complete information and study an
algorithmic problem.}

Surprisingly, Babaioff, Kleinberg, and Slivkins~\cite{BKS10} recently
showed that randomization eliminated this difficulty for a large class of
problems. They showed that, if in a single-parameter domain payments need only
be {\em truthful in expectation}, then they may be computed by solving the
original problem only once. They apply their result to Nisan and Ronen's path
auctions to get a truthful-in-expectation mechanism that uses precisely one
shortest-path computation and chooses the shortest path with probability
arbitrarily close to 1. We call this a {\em single-call mechanism}.

The usefulness of Babaioff, Kleinberg, and Slivkins' result goes far beyond
speeding up computation: {\em Their construction enables truthfulness in cases
in which computing ``$n$ versions of the original problem'' is informationally
impossible.}  To use again the Nisan-Ronen path auction, suppose that the graph
represents a packet network with existing traffic. In this case, the actual
transit times (i.e. costs to edges) may be increased by congestion. While it is
possible to estimate congestion ex ante, it is generally impossible to precisely
know its effect without transmitting a packet and explicitly measuring its
transit time. Unfortunately, since VCG prices depend on the transit times for
many different paths, na\"{\i}vely computing them will inherit any estimation
errors. Even worse, when bidders have conflicting beliefs about such errors,
{\em na\"{\i}vely computing``VCG'' prices with bad estimates may not guarantee
truthfulness} even if the errors are small enough that they not affect the path
chosen by the mechanism. In such a case, truthfulness may be regained using a
mechanism that only requires measurements along a single path, that is, a
mechanism that only requires measurements returned by a single call to the
shortest-path algorithm. We will concretely demonstrate this phenomenon later
using an example based on pay-per-click advertising auctions.

An important question arises then:  {\em In which mechanism design problems, and
to what extent, are single-call mechanisms possible?}  In this paper we study,
and largely settle, this question.  First, we show that this it is possible to
transform any mechanism that charges VCG prices in expectation into a roughly
equivalent single-call mechanism.   While similar in spirit to~\cite{BKS10}, our
reduction charges prices that are fundamentally different from the mechanism in
that paper --- they do not coincide even when applied to the same allocation
rule. Second, we give characterization theorems, delineating precisely the
single-call mechanisms that are possible, for both the VCG and single-parameter
settings. Finally, single-call constructions offer a tradeoff between
expectation and risk. Our characterization theorems allow us to derive lower
bounds on this tradeoff, establishing that our VCG construction and the
construction of~\cite{BKS10} are optimal in a general sense.

\paragraph{Mechanisms, Allocations, and Payments} One cornerstone of mechanism
design is the decomposition of a mechanism into two distinct parts: an
allocation function and a payment function. This approach has borne much fruit
--- it first revealed fundamental relationships between allocation functions and
their nearly unique truthful prices, and it subsequently allowed researchers to
study the the two problems in isolation. Like~\cite{BKS10}, we leverage this
decomposition to study payment techniques that apply to large classes of
allocation functions --- naturally, our primary requirement is that the
allocation function may only be evaluated once.

We will focus on single-call mechanisms for two classes of allocation functions
that, together, comprise most allocation functions for which truthful payments
are known: {\em monotone single-parameter functions} and {\em maximal in
distributional range (MIDR) functions.}

An allocation function is said to be {\em single-parameter} if an agent's bid
can be expressed as a single number. This setting was first studied by
Myerson~\cite{M81} in the context of single-item auctions. Subsequent
generalizations showed that truthful prices existed if and only if a
single-parameter allocation is monotone and provided an explicit
characterization of truthful payments. We will use one such characterization
developed by Archer and Tardos~\cite{AT01}.

An allocation function is said to be {\em maximal in distributional range}
(MIDR) if, for some fixed set of distributions over outcomes, the allocation
always chooses one that maximizes the social welfare of the bidders. MIDR
allocation functions are important because they are precisely the ones for which
VCG payments are truthful~\cite{DD09}.

\paragraph{Truthfulness Under Uncertainty} Our motivation for developing and
optimizing single-call mechanisms comes from scenarios where nature prohibits
computing an allocation more than once, most often due to parameter uncertainty.
We give a few examples here; more generally, we conjecture that most mechanism
design problems have similar variants.

In the uncertain shortest-path auction described earlier, truthful prices will
depend on the incremental effect of transit times adjusted for congestion. If
the auctioneer generates the network traffic, he may be able to predict the
congestion in an edge better than the edge itself and use this prediction when
computing the shortest path. However, each edge may individually disagree with
the auctioneer's estimate, and these beliefs are generally unknown to the
auctioneer. If the auctioneer were to simply compute VCG payments by combining
his estimates with players' bids, the prices would likely not be truthful. On
the other hand, we can require that payments are computed using measured transit
times instead of estimates; however, it is informationally impossible to know
the precise delay along edges that were not actually traversed. A single-call
mechanism sidesteps this hurdle by using only the delays along traversed edges
for which the delay had been precisely known.

Machine scheduling offers another application for single-call mechanisms. In
some applications (e.g. cloud services), it is common for machines to bid in
terms of cost per unit time (or other resource). It is then the responsibility
of the scheduler to estimate the time required for the job on that machine. If
the scheduler's estimates differ from a machine's belief about a job's runtime,
then we find ourselves in the same situation as the path auction --- the
standard truthful prices for this single-parameter setting will depend on
machines' beliefs about the runtimes of jobs under alternate schedules. A
single-call mechanism sidesteps this problem because it requires only the
runtimes of jobs under the schedule chosen by the mechanism, which may be
measured.

Another interesting example arises in the application of learning procedures
such as multi-arm-bandits (MABs). In recurring mechanisms, it is natural for the
auctioneer to run a learning algorithm across multiple auctions. For example,
when an online advertising auction is repeated, the auctioneer tries to learn
the likelihood that a particular ad will get clicked. Computing truthful prices
requires knowing what would have happened if the learner had been initialized
with a different set of bids. This setting was the original motivation
of\hidecite{ Babaioff, Kleinberg, and Slivkins}~\cite{BKS10}, where they showed
that their single-call construction allowed a MAB to be implemented truthfully
with $O(\sqrt T)$ regret. This contrasts with results of Babaioff, Sharma, and
Slivkins~\cite{BSS09}  and Devanur and Kakade~\cite{DK09} who showed
that any universally truthful mechanism must have regret at least
$\Omega(T^{\frac23})$ for different measurements of regret.

Finally, in Section~\ref{sec:apps} we analyze {\em single-shot pay-per-click
(PPC) advertising auctions}. A PPC advertising auction ranks bidders using their
pay-per-click bid (i.e. they only pay when they receive a click) and an estimate
of the probability of a click (the click-through rate, or CTR). If the bidders'
estimates of their own CTRs are different from the auctioneer's, truthful prices
necessarily depend on bidders' beliefs about the CTRs, which are unknown.

\paragraph{Single-Call Mechanisms and Reductions}  Our tool for creating
single-call mechanisms is {\em the single-call reduction}, the main object of
study in this paper. A {single-call reduction} is a transformation that takes an
allocation function as a black box and produces a truthful-in-expectation
mechanism that calls the allocation function once.  Since the expected payment
is equal to the truthful payment for the resulting mechanism, the payments are
dubbed {\em implicit}.

Babaioff, Kleinberg, and Slivkins~\cite{BKS10} discovered such a reduction
for single-parameter domains. Using only the guarantee that the black-box
allocation rule is monotone, their reduction produces a truthful-in-expectation
mechanism that implements the same outcome as the original allocation rule with
probability arbitrarily close to 1.\footnote{The authors of~\cite{BKS10} have
observed that their construction may be extended to any domain where the bid
space is convex.}

VCG is a mechanism design framework much broader than single-parameter.  {\em
Can we construct similar single-call mechanisms that charge VCG prices?}  We
answer this in the affirmative by giving a reduction producing, for any MIDR
allocation function, a single-call mechanism that charges VCG prices in
expectation. Analogous to~\cite{BKS10}, our reduction transforms any MIDR
allocation rule into a truthful-in-expectation mechanism that implements the
same outcome as the original allocation rule with probability arbitrarily close
to 1. However, our construction is fundamentally different in that the
distribution of payments does not coincide with~\cite{BKS10} when an allocation
is both MIDR and single-parameter. This reduction can guarantee truthfulness in
multi-parameter mechanisms with uncertainty, as described above, and can also be
used to speed up payment computation in MIDR settings like Dughmi and
Roughgarden's~\cite{DR10} truthful FPTAS for welfare-maximization packing
problems.

We next ask {\em what single-call reductions are possible}?  Babaioff et
al.\hidecite{~\cite{BKS10}} generalize to a class of self-resampling procedures.
Subsequent research~\cite{H11} generalized further (and simplified
substantially), but concisely characterizing single-call reductions remained an
open question. We give tight characterization theorems, showing that a wide
variety of reductions are possible and that payments have a very simple
characterization in both scenarios. The key technical idea is a simple proof
equating a reduction's expected payments with those required for truthfulness,
giving a sharp characterization of the parameters in the reduction. Our
technique is a very simple alternative to the contraction mapping argument
in~\cite{BKS10}.

Finally, we ask {\em what are the best single-call reductions}? As noted above,
known single-call reductions choose an outcome different from the original
allocation rule with some small probability $\delta$. The penalty for making
$\delta$ small is that the payments may occasionally be very large --- we study
this tradeoff. Our study is not unprecedented:~\cite{BKS10} asked, as an open
question, if their reduction optimized payments with respect to the welfare
loss, and Lahaie~\cite{L10} show a similar tradeoff between the size and
complexity of kernel-based payments achieving $\epsilon$-incentive compatibility
in single-call combinatorial auctions.

We study the tradeoff inherent to single-call mechanisms with respect to three
measures of expectation --- welfare, revenue, and a technical (but natural)
precision metric --- and two measures of risk --- variance and worst-case
payments. We show that our VCG reduction and the single-parameter reduction
of~\cite{BKS10} {\em simultaneously} optimize the tradeoff between expectation
and risk for all these criteria.

\section{Preliminaries}
\label{sec:prelim}

A mechanism is a protocol among $n$ rational agents that implements a social
choice function over a set of outcomes $\Ocm$. Agent $i$ has preferences over
outcomes $\ocm\in\Ocm$ given by a \emph{valuation function} $\vali: \Ocm
\rightarrow \real$. The function $\vali$ is private but is drawn from a publicly
known set $\Vali \subseteq \real^{\Ocm}$.

A {\em deterministic direct revelation mechanism} $\me$ is a social choice
function $\al: \Vali[1] \times \dots \Vali[n] \rightarrow \Ocm$, also known as
an {\em allocation rule}, and a vector of payment functions $\pricei[1], \dots,
\pricei[n]$ where $\pricei: \Vali[1] \times \dots \Vali[n] \rightarrow \real$
is the amount that agent $i$ pays to the mechanism designer.  When a direct
revelation\footnote{``Direct revelation'' means that an agent's bid $b_i$ is an
element of $V_i$. In general this need not be the case; however, by the
revelation principle, any social choice rule that may be truthfully implemented
may be implemented as a direct revelation mechanism that charges the same
payments in equilibrium.} mechanism is instantiated, each agent reports a {\em
bid} $b_i\in\Vali$. The mechanism uses bids $b = (b_1,\dots,b_n)$ to choose an
outcome $\al(b)\in\Ocm$ and to compute payments $\pricei(b)$.  The utility
$u_i(\vali,\ocm)$ that agent $i$ receives is $u_i(\vali,\ocm) = \vali(\ocm) -
\pricei$.  A mechanism is {\em truthful} (or incentive compatible) if bidding
truthfully (i.e. $b_i=\vali$) is a dominant strategy. Formally, for each $i$,
each $\vali[-i] \in \Vali[-i]$, and every $\vali,\vali' \in \Vali$, we have
$u_i(\vali,A(\val)) \geq u_i(\vali,A(\vali',\val_{-i}))$, where $\val_{-i}$
denotes the vector of valuations for all agents except agent $i$.

A mechanism is {\em ex-post individually rational} (IR) if agents always get
non-negative utility, and mechanism has {\em no positive transfers} (NPT) if for
each agent $i$ and each $\val \in \Val$, $\pricei(\val) \geq 0$, i.e., the
mechanism never pays a player money.

A randomized mechanism is a distribution over deterministic mechanisms. Thus,
$\al(b)$ and $\pricei(b)$ are random variables. For randomized mechanisms,
properties like truthfulness may be said to hold universally or in expectation.
A randomized mechanism is {\em universally truthful} if it is truthful for every
deterministic mechanism in its support. It is {\em truthful in expectation} if,
in expectation over the randomization of the mechanism, truthful bidding is a
dominant strategy.  Henceforth, we use truthful, IR, and NPT to mean truthful in
expectation unless otherwise noted.

\paragraph{MIDR Allocation Rules}

MIDR mechanisms are variants of {\em VCG mechanisms}, mechanisms that maximize
social welfare and charge ``VCG payments''. Formally, a VCG mechanism's social
choice rule satisfies $\al(\val) \in$ $\underset{\ocm \in \Ocm}{\mathrm{argmax}}
\sum_j \valj(\ocm)$, and its payments are $\pricei(\val) = h_i(\val_{-i}) -
\sum_{j \neq i}\valj(\al(\val))$ for some function $h_i:\Vali[-i] \rightarrow
\Re$.  VCG payments are the only universal technique known to induce truthful
bidding.  The most common implementation of VCG payments uses the Clarke-Pivot
payment rule: set $h_i(\val_{-i}) = \underset{\ocm \in \Ocm}{\max}(\sum_{j \neq
i}\valj(\ocm))$, which gives the only payments that simultaneously satisfy
truthfulness, IR, and NPT.

More generally, any allocation rule that maximizes an affine function of agents'
valuations can be truthfully implemented with VCG payments.  Moreover, Roberts'
theorem~\cite{Rob79} implies that in a general setting (when $\Vali =
\real^{\Ocm}$), if $\al$ is onto (every outcome can be realized), then $\al$ has
truthful payments if and only if it is an affine maximizer.  If the ``onto''
restriction is relaxed, a social choice function is truthfully implementable
with VCG payments if and only if it is (weighted) maximal-in-range
(MIR)~\cite{NR00} or, for randomized mechanisms, maximal-in-distributional-range
(MIDR)~\cite{DD09}:

\begin{definition} An allocation rule $\al$ is {\em MIDR} if there is a set
$\Dist$ of probability distributions over outcomes such that $\al$ outputs a
random sample from the distribution $\dist \in \Dist$ that maximizes expected
welfare. Formally, for each $\val \in \Val$, $\al(\val) = \ocm \sim \dist^*$
where $\dist^* \in \underset{\dist \in \Dist}{\mathrm{argmax}}\ \ex_{\ocm \sim
\dist}[\sum_i \vali(\ocm)]$.  \end{definition} A weighted MIDR allocation rule
maximizes the weighted social welfare $\sum_i w_i\vali(\ocm)$ for $w_i\geq0$.

\paragraph{Single-Parameter Domains} A larger class of social choice rules can
be implemented when $\Vali$ is single dimensional. We say that a social choice
rule has a single-parameter domain if $\vali(\ocm) = \tyi\pfn(\ocm)$ for some
publicly known function $f_i:\Ocm\rightarrow\Re_+$. The value $\tyi \in \Tyi$ is
an agent's type ($\Tyi$ is her type-space, and $\Ty=\Ty_1 \times \dots \times
\Ty_n$), and submitting $i$'s bid precisely requires stating $b_i=\tyi$. When
$\Ty=\Re^n_+$, we say that bidders have {\em positive types}. We also use
$\al_{i}(b)=\pfn(\al(b))$ as shorthand, and we say $\al$ is bounded if the
functions $\al_i$ are bounded functions.

A single-parameter social choice rule may be implemented if and only if it is
{\em monotone}, where $\al:\Ty \rightarrow \Ocm$ is said to be monotone if for
each agent $i$, for all $b_{-i} \in \Ty_{-i}$ and for every two bids $b_i \geq
b_i'$, we have  $\al_i(b_i, b_{-i}) \geq \al_i(b_i', b_{-i})$.  This was first
shown for a single item auction by Myerson~\shortcite{M81}; Archer and
Tardos~\shortcite{AT01} gave the current generalization:

\begin{theorem}\label{thm:MAT}[Myerson + Archer-Tardos] For a single parameter
domain, an allocation rule $\al$ has truthful payments $(\pricei[1], \dots,
\pricei[n])$ if and only if $\al$ is monotone. These payments take the form
\[\pricei(b) = h_i(b_{-i}) + b_i\al_{i}(b_i,b_{-i})-
\int_{0}^{b_i}\al_i(u,b_{-i})\,du, \] where $h_i(b_{-i})$ is independent of
$b_i$.  \end{theorem} These payments simultaneously satisfy IR and NPT if and
only if $\pricei^0(b_{-i}) = 0$. Such a mechanism is said to be normalized.

\section{Single-call mechanisms}
\label{sec:sc}
We call a mechanism a single-call mechanism if it only evaluates the allocation
function once:

\begin{definition} A {\em single-call mechanism} $\me$ for an allocation rule
$\al$ is a truthful mechanism that has only oracle access to $\al$ and computes
both the allocation and payments with a single call to $\al$.
\end{definition}

To construct a single-call mechanism, we must first specify the possible
allocation functions $\al$ and then construct one procedure that yields a
single-call mechanism for any $\al$ in this set. Thus, the tool for creating a
single-call mechanism is a single-call reduction: 
\begin{definition} A {\em
single-call reduction} is a procedure that takes any allocation function $\al$
from a fixed set (as a black box) and returns a single-call mechanism.
\end{definition} 
For example, the procedure of~\cite{BKS10} is a single-call
reduction that takes any $\al$ drawn from the set of all monotone, bounded,
single-parameter allocation rules and returns a single-call mechanism.
Similarly, our construction for VCG prices is a single-call reduction that takes
any $\al$ that is MIDR and returns a single-call mechanism.

To formalize single-call reductions, we first note the following requirements:
\begin{itemize}
\item A reduction must
take a bid vector $b$ and a black-box allocation function $\al$ as input.
\item A reduction must evaluate $\al$ on at most one bid vector $\hat b$,
causing the outcome $\al(\hat b)$ to be realized.\footnote{Strictly speaking,
there may be settings where a single-call reduction could realize an outcome
other than $\al(\hat b)$. However, our restriction follows naturally in
scenarios where ``computing $\al(b)$'' means realizing $\al(b)$ and making
measurements. It is also required for complete generality because there is no
reason to believe that the designer knows how to realize any outcome other than
$\al(\hat b)$.} \item A reduction must charge payments $\lambda_i$ that are a
function of  $b$, $\hat b$, and $\al(\hat b)$ (and possibly its own
randomness).
\end{itemize}
These requirements suggest the following generic definition of a single-call reduction
 to turn an allocation function $\al$ into a truthful-in-expectation single-call mechanism $\me=(\asc,\{\psci\})$:
\begin{enumerate} \item Solicit the bid vector $b$ from agents. 

\item Use $b$ to compute the modified bid vector $\hat b$. This implicitly defines a
probability measure $\mu_b(B)$ denoting the probability of choosing $\hat b\in
B\subseteq\Vali[1]\times \cdots \times \Vali[n]$ as the modified (resampled) bid vector when $b$ is the actual bid vector. 
When $\hat b_i\neq b_i$, we say that $i$'s bid was resampled.  

\item Declare the outcome to be $\al(\hat b)$, i.e. evaluate $\al$ at the
modified bid vector $\hat b$. This implicitly defines the allocation function
$\asc(b)$ which samples $\hat b\sim \mu_b$ and chooses the outcome $\al(\hat
b)$.  The resampling procedure must ensure that truthful payments $\psc(b)$ exist for 
$\asc(b)$; Note that $\asc(b)$ and $\psc(b)$ are random variables that depend on
the randomly resampled bid vector $\hat b$. Also, $\asc(b)$ and $\psc(b)$ are randomized 
even if $\al(b)$ and $\price(b)$ are deterministic;

\item Use $b$, $\hat b$, and $\al(\hat b)$ to compute payments
$\lambda_i(\al(\hat b),\hat b,b)$ that satisfy truthfulness in expectation, that is, charge player $i$ a payment
$\lambda_i(\al(\hat b),\hat b,b)$ such that $\underset{\hat b}{\expect}[\lambda_i(A(\hat
b), \hat b,b)] = \underset{\hat b}{\expect}[\psci(b)]$. 
\end{enumerate}
This general procedure is illustrated in Algorithm~\ref{alg:gsc}.

\begin{algorithm}[tb]\label{alg:gsc}
\SetKwInOut{Input}{input}\SetKwInOut{Output}{output}
\Input{Black box access to an allocation function $\al$, which is drawn from a known set.}
\Output{Truthful-in-expectation mechanism $\me=(\asc,\{\psci\})$.}
\BlankLine
\nl Solicit bid vector $b$ from agents\;
\nl Sample $\hat b\sim\mub$\;
\nl Realize the outcome $\al(\hat b)$\tcp*[r]{$\asc(b)$ is the random function $\al(\hat b)$ where $\hat b\sim\mub$}
\nl Charge payments $\lambda(\al(\hat b),\hat b,b)$\tcp*[r]{$\psci(b)$ is the random function $\lambda_i(\al(\hat b),\hat b,b)$ where $\hat b\sim\mub$}
\caption{Generic Single-Call Reduction $(\mu,\{\lambda_i\})$}
\end{algorithm}

We describe a single-call reduction in the above framework by the tuple
$(\mu,\{\lambda_i\})$, where $\mu$ implies specifying the resampling measure
$\mu_b$ for all $b\in\Vali[1] \times \cdots \times \Vali[n]$.  Since payments should be finite, we require that
$\lambda_i$ be finite everywhere, and we also require that it be integrable. For
the rest of this paper, we assume that $\lambda_i$'s are deterministic. For
randomized $\lambda_i$'s, the characterization theorems still hold with
$\lambda_i$'s replaced by their expectations over the randomness used. 

We say that a reduction is {\em normalized} if $b_i(\al(b))=0$ for all $i$ implies
$\lambda_i(\al(\hat b),\hat b,b)=0$, i.e. when every agent receives zero value,
all payments are zero.

\subsection{Optimal Reductions --- Expectation vs. Risk}

There are two downsides to the mechanisms produced by single-call reductions.
First, there is a penalty in {\em expectation}, i.e., the expected outcome 
$\ex_{\hat b}[\al(\hat b)]$ produced by the reduction is not identical to the 
desired outcome, $\al(b)$. This modified outcome may reduce the expected welfare or revenue of the mechanism, or
it may simply cause it to do the ``wrong'' thing.

Second, there is a penalty in {\em risk} because the payments
$\lambda$ may vary significantly, i.e. for a fixed $b$ the payments at
different resampled bids $\hat b$ could be very different.  In particular, the
magnitude of the payment charged by the single-call mechanism may be much larger
than the payments in the original mechanism, i.e. it may be that
$|\lambda_i|\gg|\pricei|$ for certain outcomes.

Our characterization theorems reveal that there is a fundamental trade-off between
expectation and risk. Thus, we call a reduction optimal if it minimizes risk with respect
to a lower bound on the expectation.

\subsubsection{Expectation} We study three criteria for measuring the expectation of
a reduction: $\Pr(\hat b=b|b)$,
social welfare, and revenue.

The first criterion, $\Pr(\hat b=b|b)$ (the precision), measures the likelihood that
the reduction modifies players' bids. This criterion is natural when modifying
bids is inherently undesirable:
\begin{definition}
The {\em precision} of a reduction $\alphap$ is the probability that the reduction does not alter any player's bid:
\[\alphap\equiv\min_b\Pr(\hat b=b|b)\enspace.\]
\end{definition}

The other criteria measure standard quantities in mechanism design:

\begin{definition}
The {\em welfare approximation} $\alphaw$ of a single-call reduction is given by
the worst-case ratio between the welfare of the single-call mechanism and the
welfare of the original allocation function:
\[\alphaw=\min_{\al,b}\frac{\ex_{\hat
b}\left[\sum_i b_i(\asci(b))\right]}{\sum_ib_i(\al_i(b))}\enspace.\]
When the welfare of $\al$ is zero, $\alphaw=1$ if the welfare of $\asc$ is also zero and unbounded otherwise.
\end{definition}

\begin{definition}
The {\em revenue approximation} $\alphar$ of a single-call reduction is given by
the worst-case ratio between the revenue of the single-call mechanism and the
revenue of the original allocation function:
\[\alphar=\min_{\al,b}\frac{\ex_{\hat b}\left[\sum_i\psci(b)\right]}{\sum_i\pricei(b)}\enspace.\]
When the revenue of $\al$ is zero, then $\alphar=1$ when the revenue of $\asc$ is also zero and unbounded otherwise.
\end{definition}

In the case of continuous spaces we replace $\min$/$\max$ with $\inf$/$\sup$ as appropriate for infinite domains.

\subsubsection{Risk} We measure risk through both the variance of payments and their worst-case
magnitude.\footnote{Intuition suggests optimizing with respect to a high-probability bound. Unfortunately,
this is problematic because ignoring low-probability events can dramatically change the expected
payment. Thus, in general it is not reasonable to conclude a priori that low-probability events can be ignored.}
In order to make a meaningful comparison across different allocation functions and bids, we normalize by players'
bids:\footnote{Intuition also suggests normalizing by the truthful prices for $\al$ (i.e. by $\pricei$), but constant
allocation functions such as $\al_i(b)=1$ have $\pricei=0$, making this impossible. Bid-normalized payments
are a next logical choice.}
\begin{definition} Decompose $\lambda_i$ into terms which depend only on the payoff to a single bidder $j$ (i.e. on $b_j(\al(\hat b))$ instead of $\al(\hat b)$):
\[\lambda_i(\al(\hat b),\hat b,b) = \sum_j \lambda_{ij}(b_j(\al(\hat b)),\hat b,b)\]
(our characterizations in Sections~\ref{sec:vcg-char} and~\ref{sec:sp-char} show that this is possible for
our settings). Then the
{\em bid-normalized payments} of the reduction are given by
\[\sum_j \frac{\lambda_{ij}(b_j(\al(\hat
b)),\hat b,b)}{b_j(\al(\hat b))}\enspace.\]
\end{definition}
We can thus write the variance of bid-normalized payments as
\[\max_{\al,i}\var_{\hat b\sim\mub}\left(\sum_j \frac{\lambda_{ij}(b_j(\al(\hat
b)),\hat b,b)}{b_j(\al(\hat b))}\right)\]
and the worst-case magnitude as
\[\max_{\al,i,\hat b}\left|\sum_j \frac{\lambda_{ij}(b_j(\al(\hat
b)),\hat b,b)}{b_j(\al(\hat b))}\right|\]
where we replace $\min$/$\max$ with $\inf$/$\sup$ as appropriate for infinite domains.

\subsubsection{Optimality} We define an optimal reduction as one that
{\em simultaneously} optimizes the six-way trade-off between expectation and risk:
\begin{definition}
A single-call reduction optimizes the {\em variance of/worst-case} payments with respect to
{\em precision/welfare/revenue} for a set of
allocation functions if for every bid $b$, it minimizes the variance of/worst-case
normalized payments over all possible reductions that achieve
a precision of $\alphap$ / welfare approximation of $\alphaw$ / revenue
approximation of $\alphar$.
\end{definition}

\section{Maximal-in-distributional-range reductions}
\label{sec:vcg-char}

In this section, we show how to construct a single-call reduction for MIDR
allocation rules, i.e. we show how to construct a randomized, truthful
mechanism from an arbitrary MIDR allocation rule $\al$ using only a single
black-box call to $\al$. The main results are Theorem \ref{thm:vcg-char}, a
characterization of all reductions that use VCG payments for an arbitrary MIDR
allocation rule, and an explicit construction that optimizes the expectation-risk
tradeoff.

Truthful payments for MIDR allocation rules are given by VCG payments with the
Clarke-Pivot rule:\footnote{If we relax the no positive transfers requirement, a
trivial way to construct a single-call mechanism is to ignore the first term in~\eqref{eqn:vcg_payment}. However, the resulting mechanism would
make a huge loss because no agent would ever pay the mechanism.}
\begin{eqnarray}\label{eqn:vcg_payment}
\ex[p_i]&=&\ex[\mbox{total welfare of bidders without
$i$}]-\ex[\mbox{total welfare of bidders $j\neq i$ with $i$}]
\end{eqnarray}
 
(where the expectation is over the randomization in the given MIDR allocation
rule).  The reduction comes from this formula for $\ex[p_i]$: we need to measure
the welfare without agent $i$ (the first term in the RHS), so, with some
probability, we ignore agent $i$ and maximize the welfare of the remaining
agents. Intuitively, this is equivalent to evaluating the allocation function
where $i$'s bid is changed to a ``zero'' bid while other bids remain the same.

Unfortunately, having removed agent $i$, even with a small probability, means
that computing truthful payments for agent $j\neq i$ requires knowing the
allocation where both $i$ and $j$ are ignored. By induction, a single-call
mechanism must generate all sets of agents $M\subseteq[n]$ with some
probability.  Thus, we get an intuitive picture of the reduction's behavior: it
will randomly pick a set of bidders $M\subseteq[n]$ and zero the bids of agents
not in $M$.

\subsection{Characterizing Truthfulness} We consider reductions in which
$i$'s resampled bid $\hat b_i$ is always $b_i$ or zero,\footnote{Even if explicit ``zero'' bids are not known to the reduction, we assume that the reduction can induce $\al$ to optimize the utility of an arbitrary
subset of agents. Note that a black-box allocation function can only be turned
into a truthful mechanism (even if multiple calls to $\al$ are allowed) if it
can ignore at least one bidder at a time, so our assumption is not unreasonable.}
where ``zero'' means that the agent has a valuation of zero for all outcomes. That is, the resampling measure $\mu_b(B)$ represents a discrete distribution
over the bids $\{\hat b^M\}$ where $M\subseteq[n]$ is a set of agents and
\begin{equation*}
\hat b_i^M=\begin{cases}
b_i&i\in M\\
0&i\not\in M
\end{cases}
\end{equation*}
Resampling to $\hat b^M$ is equivalent to ignoring the welfare of agents
outside $M$ and evaluating $A$ at $b$.

In the most general setting, our restriction to zeroing reductions is without 
loss of generality because $b$ and
zero are the only bids that are guaranteed to be valid inputs to $\al$ for all
MIDR allocation functions $\al$. That said, even if a
multi-parameter bid structure were known, VCG payments do not depend on the outcome at any other
bid. Thus, intuition suggests that resampling to other bids will not be helpful even if it is possible. This intuition can be formalized, but we do not do it here.

Let $\pi(M)$ be a distribution over sets $M\subseteq [n]$. 
We define the \emph{associated coefficients}
$c_i^\pi(M)$ as:
\begin{equation*}\label{eqn:midr_associated_coeff}
c_i^\pi(M)=\begin{cases}
-1,&i\in M\\
\frac{\pi(M\cup\{i\})}{\pi(M)},&i\not\in M
\end{cases}
\end{equation*}
Intuitively, $c_i^\pi$ is the weighting that ensures $-\pi(M\cup\{i\})c_i^\pi(M\cup\{i\})=\pi(M)c_i^\pi(M)$ (where $i\not\in M$) to
match the terms in~\eqref{eqn:vcg_payment}.

We prove the following characterization of all truthful MIDR reductions $(\pi,\{\lambda_i\})$ that work for all MIDR $\al$:

\begin{theorem} \label{thm:vcg-char} 
A normalized single-call reduction, with VCG payments, for the set of all MIDR
allocation rules satisfies truthfulness, individual rationality, and no positive
transfers in an ex-post sense if and only if it takes the form
$(\pi,\{\lambda_i\})$ where $\pi(M)$ is a distribution over sets
$M\subseteq[n]$, the coefficients $c_i^{\pi(M)}$ are finite, and payments take
the form $$\lambda_i(\al(\hat b^M),\hat b^M,b)=c_i^\pi(M)\sum_{j\neq
i}b_j(\al(\hat
b^M))\enspace.$$ 
\end{theorem}

\begin{proof} 
Recall that in general, a multi-parameter allocation function that can be rendered
truthful by VCG payments must be MIDR. Thus, our reduction must ensure that
$\asc$ is MIDR, and we first derive the implications of this requirement on the
single-call reduction. We have already assumed that $\mu_b(B)$ is a distribution over
bids $\{\hat b^M\}$. Let $\pi_b(M)$ be the probability of selecting $\hat b^M$ given $b$.

First, we show that $\asc$ is always MIDR if and only if $\pi_b(M)$ does not depend on $b$. For the
{\em if} direction, if $\pi_b(M)$ is independent of $b$ then $\asc$ is a distribution over MIDR allocation rules, and
by~\cite{DR10}, such an allocation rule is MIDR.

For the {\em only if} direction, we use contradiction. Assume that there are some bids $x$ and $y$
such that $\pi_x(M)\neq\pi_y(M)$ for some $M$. Then there exists a set $S\subseteq[n]$ such that
$\pr_\pi(M\subseteq S|x)\neq\pr_\pi(M\subseteq S|y)$ (by contradiction and induction, start with $S=\emptyset$).
Consider an allocation function that has welfare $\sum_ib_i(\al(\hat b^M))=0$ for $M\subseteq S$ and
$\sum_ib_i(\al(\hat b^M))=1$ otherwise. The welfare of $\asc$ will be precisely $1-\pr_\pi(M\subseteq S)$,
implying that for either $x$ or $y$, $\asc$ did not chose the distribution that maximized social welfare
and is therefore not MIDR. Thus, the allocation rule $\asc$ is MIDR for all MIDR $\al$ if and only if $\mu_b(B)$ is
a discrete distribution $\pi(M)$ independent of $b$.

Next, we write VCG payments for $\asc$ that satisfy individual rationality 
and no positive transfers using the Clarke-Pivot payment rule:
\begin{eqnarray}
\ex[\psci]&=&\sum_{j\neq i}\sum_{M\subseteq[n]}\pi(M)
b_j(\al(\hat b^{M\setminus\{i\}}))-\sum_{j\neq i}\sum_{M\subseteq[n]}\pi(M) b_j(\al(\hat b^{M}))\nonumber\\
&=&\sum_{M|i\not\in M}\pi(M\cup\{i\})\sum_{j\neq i}b_j(\al(\hat
b^M))\TCSeafeed-\sum_{M|i\in M}\pi(M)\sum_{j\neq i}b_j(\al(\hat b^M))\enspace.\label{eqn:vcg-char-ep1}
\end{eqnarray}

By definition of $\lambda_i(A(\hat b^M),\hat b^M,b)$, we know that the expected
payment made by $i$ will be
\begin{equation}
\label{eqn:vcg-char-ep2}\ex[\psci]=\sum_{M\subseteq[n]}\pi(M)\lambda_i(A(\hat b^M),\hat b^M,b)\enspace.
\end{equation}
The two formulas for payments in~\eqref{eqn:vcg-char-ep1} and~\eqref{eqn:vcg-char-ep2} 
must be equal:
\begin{eqnarray*} \sum_{M\subseteq[n]}\pi(M)\lambda_i(A(\hat b^M),\hat
b^M,b)\TCSeas&=&\sum_{M|i\not\in M}\pi(M\cup\{i\})\sum_{j\neq i}b_j(\al(\hat
b^M))\TCSeafeed-\sum_{M|i\in M}\pi(M)\sum_{j\neq i}b_j(\al(\hat b^M))\enspace.
\end{eqnarray*}
Since $\al$ may be any MIDR allocation function, the only way this can hold is when terms corresponding to 
each $M$ are equal, i.e., for all $i,\ M$
\begin{equation}\label{eqn:vcgPaymentEquation}
\pi(M)\lambda_i(A(\hat b^M),\hat b^M,b)=\begin{cases}\pi(M\cup\{i\})\sum_{j\neq
i}b_j(\al(\hat b^M),&i\not\in M\\
-\pi(M)\sum_{j\neq i}b_j(\al(\hat b^M))&i\in M\enspace.
\end{cases}
\end{equation}
To see that this is necessary, construct two allocation functions $\al$ and
$\al'$ such that $b_j(\al(\hat b^M))=b_j(\al'(\hat b^M))$ for all $M\neq\bar M$ and
$b_j(\al(\hat b^{\bar M}))=0$. It immediately follows that if the reduction works
for both $\al$ and $\al'$, then~\eqref{eqn:vcgPaymentEquation} must hold for $\bar M$ under
$\al$. Since $\bar M$ is arbitrary, it follows that~\eqref{eqn:vcgPaymentEquation} must hold for all $M$.

The theorem immediately follows from the above equality. 
\end{proof}

\begin{remark} Note that this theorem forbids some distributions $\pi(M)$ from
being used to construct a single-call reduction --- in particular, it requires
that $\pi(M)>0$ for all $M\subseteq[n]$, otherwise some payment
$\lambda_i(\cdot)$ will be infinite for nontrivial allocation rules.
For example, an obviously forbidden distribution
is the one that never changes bids, i.e. the one with $\pi([n])=1$. This matches
the intuition that a single-call mechanism must occasionally modify bids.
\end{remark}

\subsection{A Single-Call MIDR Reduction}
We now give an explicit single-call reduction for MIDR allocation functions. Our reduction
$\optMidrRedn(\al,\gamma)$ (illustrated in Algorithm~\ref{alg:vcgsc}) is defined by the following resampling distribution $\pibar$ 
parameterized by a constant $\gamma\in(0,1)$:
\begin{equation}\label{eqn:piBarDefn}
\pibar(M)=\gamma^{n-|M|}(1-\gamma)^{|M|}
\end{equation}
That is, each agent $i$ is independently dropped from $M$ with probability
$\gamma$. Thus sampling from the distribution $\pibar$ is computationally easy. Following Theorem \ref{thm:vcg-char}, we charge payments
$\lambda_i(\al(\hat b^M),\hat b^M,b)=c_i^{\pibar}(M)\sum_{j\neq i}b_i(\al(\hat b^M))$
where
\begin{equation*}
c_i^{\pibar}(M)=\begin{cases}
-1,&i\in M\\
\frac{1-\gamma}{\gamma},&i\not\in M
\end{cases}
\end{equation*}

\begin{corollary}[of Theorem~\ref{thm:vcg-char}]\label{cor:vcg-redn}
The mechanism
\[\me=(\asc,\{\psci\})=\optMidrRedn(\al,\gamma)\]
calls $\al$ once
and it satisfies truthfulness, individual rationality, and no positive transfers in an ex-post sense for all MIDR $\al$.
\end{corollary}

\begin{algorithm}[tb]\label{alg:vcgsc}
\SetKwInOut{Input}{input}\SetKwInOut{Output}{output}
\SetKwIF{WP}{ElseWP}{Otherwise}{with probability}{}{with probability}{otherwise}{end}
\Input{MIDR allocation function $\al$.}
\Output{Truthful-in-expectation mechanism $\me=(\asc,\{\psci\})$.}
\BlankLine
\nl Solicit bids $b$ from agents\;
\nl \For{$i\in[n]$}{
 \eWP {$1-\gamma$}{
   Add agent $i$ to set $M$\;}{
   Drop agent $i$ from $M$\;}}
\nl Realize the outcome $\al(\hat b^M)$\;
\nl Charge payments\\
$\lambda_i(\al(\hat b^M),\hat b^M,b)=\left(\sum_{j\neq i}b_j(\al(\hat b^M))\right)\times\begin{cases}-1,&i\in M\\\frac{1-\gamma}{\gamma},&i\not\in M\end{cases}$\;
\caption{$\optMidrRedn(\al,\gamma)$ --- A single-call reduction for MIDR allocation functions}
\end{algorithm}

\subsection{Optimal Single-Call MIDR Reductions}
\label{sec:vcg-opt}

\label{sec:vcg-char-opt}

We now prove that the construction $\optMidrRedn(\al,\gamma)$ is optimal for the definitions of optimality given in Section~\ref{sec:sc}. Theorem~\ref{thm:vcg-char}
implies that the bid-normalized payments will be
\[\sum_j \frac{\lambda_{ij}(b_j(\al(\hat b)),\hat b,
b)}{b_j(\al(\hat b))}=(n-1)c_i^\pi(M)\]
Thus, it is sufficient to optimize the variance as $\max_i\var_{M\sim\pi} c_i^\pi(M)$ and the worst-case as $\max_{i,M}|c_i^\pi(M)|$.
\subsubsection{Optimizing Risk vs. Precision} 

\begin{theorem}
\label{thm:vcg-opt}
The reduction $\optMidrRedn(\al,\gamma)$ uniquely minimizes both the payment variance
 and the worst-case payment among all reductions that achieve a precision of at least $\alphap=(1-\gamma)^n$.

That is, for any other distribution $\pi$ with precision $\pi([n])\geq(1-\gamma)^n$, the payment
variance is larger, i.e.
\[\max_i\var_{M\sim\pi} c_i^\pi(M)>\max_i\var_{M\sim\pibar} c_i^{\pibar}(M)\enspace,\]
and the worst-case payment is larger, i.e.
\[\max_{i,M}|c_i^\pi(M)|>\max_{i,M}|c_i^{\pibar}(M)|\enspace.\]
\end{theorem}
\begin{proof} First we prove optimality for the worst-case payment $\max_{i,M}|c_i^\pi(M)|$ by contradiction. Assume that some distribution $\pi(M)$ does as well as $\pibar(M)$. 
Then it must be that $\max_{i,M} c_i^\pi(M)\leq\max_{i,M}c_i^{\pibar}(M)$ (the largest coefficient is not bigger), 
and $\pi([n])\geq\pibar([n]) = \alphap$ (it respects the lower bound on precision). 
Since $\max c_i^{\pibar}(M)=\frac{1-\gamma}{\gamma}$, it must be that for all $M$ and $i\not\in M$,
\begin{equation*}
\frac{\pi(M\cup\{i\})}{\pi(M)}\leq\max_{i,M}c_i^{\pibar}(M)=\frac{1-\gamma}{\gamma}=\frac{\pibar(M\cup\{i\})}{\pibar(M)}\enspace.
\end{equation*}
Therefore, for any bidder $i$, it must be that
\begin{equation*}
\frac{\pi([n])}{\pi([n]\setminus\{i\})}\leq\frac{\pibar([n])}{\pibar([n]\setminus\{i\})}\enspace.
\end{equation*}
Since $\pi([n])\geq\pibar([n])$, it follows that $\pi([n]\setminus\{i\})\geq\pibar([n]\setminus\{i\})$. Repeating this argument, it follows by induction that $\pi(M)\geq\pibar(M)$ for any set $M$.

However, we also know that both $\pi(M)$ and $\pibar(M)$ are distributions so both have
to sum to one over all $M$. Given that $\pi(M)\geq\pibar(M)$ for all $M$, this implies $\pi(M)=\pibar(M)$.
Thus, $\pibar(M)$ is uniquely optimal.

Second, we argue that $\pibar$ optimizes the payment variance. The variance of bidder $i$'s payments is
\begin{align*}
\var_{M\sim\pi}c_i^\pi(M)&=\sum_{M\subseteq[n]}\pi(M)\left(c_i^\pi(M)\right)^2-\left(\sum_{M\subseteq[n]}\pi(M)c_i^\pi(M)\right)^2\\
&=\sum_{M\subseteq[n]}\pi(M)\left(c_i^\pi(M)\right)^2-0\\
&=\sum_{M\subseteq[n]\setminus\{i\}}(\pi(M)+\pi(M\cup\{i\})\frac{\pi(M\cup\{i\})}{\pi(M)}
\end{align*}
This is minimized when $\pr(i\in M)$ is independent of other bidders (Lemma~\ref{lem:babaioff-hfunc2}), i.e.
$\frac{\pi(M\cup\{i\})}{\pi(M)}=\frac{1-\gamma_i}{\gamma_i}$ for some constant $\gamma_i$. For such a distribution,
the precision will be
\[\pi([n])=\prod_i(1-\gamma_i)\enspace.\]
It follows that the maximum variance is $\max_i\frac{1-\gamma_i}{\gamma_i}$, and it will only be minimized when $\gamma_i=\gamma_j$ for all $i\neq j$, which corresponds precisely
to the distribution $\pibar$. 
\end{proof}

\subsubsection{Optimizing Risk vs. Welfare} 

A natural optimization metric is the social welfare of $\asc$ (indeed, this
was an open question from\hidecite{ Babaioff, Kleinberg, and Slivkins}~\cite{BKS10} in
the single-parameter setting).

Unfortunately, since MIDR allocation rules may generate negative utilities and
remain MIDR under additive shifts of the valuation function, one can make the
welfare approximation arbitrarily bad (indeed, even undefined) by subtracting a
constant from each player's valuation. Thus, if valuation functions may be
negative, we cannot meaningfully optimize the loss in social welfare.

However, when valuation functions are known to be nonnegative, then the following lemma shows that the
worst-case welfare approximation is bounded: 
\begin{lemma}\label{lem:vcg-swapx}
The reduction $\optMidrRedn(\al,\gamma)$ obtains an
$\alphaw=\min_i\Pr_\pi(i\in M)=1-\gamma$ approximation to the social welfare, and
there is an allocation function $\al$ and bid $b$ such that this bound is
tight.  
\end{lemma}

The idea for the lower bound is that the sum of welfare of bidders in $M$ cannot be 
lower at $\al(\hat b^M)$ than at $\al(\hat b^{[n]})$ because that would imply
$\al$ did not maximize the social welfare of bidders in $M$ at $\hat b^M$. The
worst case scenario occurs when one player receives all the welfare. The proof
is given in Appendix~\ref{sec:a-midr-opt}.

Using this lemma, we can show that $\optMidrRedn(\al,\gamma)$ is optimal:
\begin{theorem}\label{thm:vcg-swopt}
The reduction $\optMidrRedn(\al,\gamma)$ minimizes payment variance
and worst-case payments among all reductions that achieve a welfare approximation of at least $\alphaw=1-\gamma$.
\end{theorem}

The proof is given in
Appendix~\ref{sec:a-midr-opt}.

\subsubsection{Optimizing Risk vs. Revenue}

The following lemma implies that a lower bound on 
the factor of approximation to revenue is equivalent to a lower bound on
precision. 
\begin{lemma}\label{lem:vcg-rapx} 
The reduction $\optMidrRedn(\al,\gamma)$ obtains an $\alpha_\pi=\pi([n])=(1-\gamma)^n$ approximation to the revenue, and this is tight.
\end{lemma}

Since Theorem~\ref{thm:vcg-opt} says that $\optMidrRedn(\al,\gamma)$ optimizes payments
with respect to precision, it similarly follows that it optimizes payments with respect to revenue:

\begin{theorem}\label{thm:vcg-ropt} The reduction $\optMidrRedn(\al,\gamma)$
minimizes payment variance and the worst-case payment among all reductions that guarantee
an $\alphar=(1-\gamma)^n$ approximation to revenue.
\end{theorem}

\section{A Single-call application --- PPC AdAuctions}
\label{sec:apps}

Pay-per-click (PPC) AdAuctions are a prime example of mechanisms in which uncertainty
can destroy truthfulness. There is a deep literature on truthful ad auctions, much of
which makes a powerful assumption: the likelihood that a user clicks in any given 
setting is a commonly-held belief.
In reality, this simply is not true. Auctioneers make their best effort to estimate the likelihood
of a click; however,
anecdotal evidence~\cite{J10} suggests that advertisers manipulate their bids according to the perceived accuracy of the
auctioneer's estimates. As we will illustrate in this section, even if the auctioneer's
estimates are good enough to (say) maximize welfare given the current bids, they are not
sufficient to compute truthful prices. We show that single-call mechanisms can
recover truthfulness in PPC ad auctions in spite of these conflicting beliefs.

In a standard PPC ad auction, $n$ advertisers compete for $m\ll n$ slots. The value to
an advertiser depends on the likelihood of a click, called the click-through-rate (CTR) $c$,
and the value to the advertiser once the user has clicked, the value-per-click $v$. The
expected value to an advertiser is thus $cv$. The auctioneer's job is to assign advertisers
to slots and compute per-click payments --- bidders are only charged when a click occurs.
Both tasks require knowing the CTRs for common objectives
like welfare or revenue maximization, so the auctioneer
must also maintain estimates of the CTRs, which we denote by $c'$.

Researchers generally acknowledge that, in reality,
both $c$ and $v$ may depend arbitrarily on the outcome --- they certainly depend on
the quality and relevance of the particular ad being shown, but they also depend on where
the ad is shown and on which other ads are shown nearby. However, for analytical tractability, the
parameters $c$ and $v$ are often assumed to have a very restricted structure. We discuss
two different structures to illustrate the pervasiveness of the problem caused by estimation error
and to show how different single-call reductions may be applied.

\paragraph{Outcome-Independent Values and Separable CTRs} In the ad auction literature,
it is common to assume that a bidder's value-per-click $v_i$ is independent of the assignment and
that the CTR is separable, that is, it takes the form $c=\alpha_j\beta_i$, where $\beta_i$ depends only on the ad
and $\alpha_j$ depends only on the slot $j\in[m]$ where the ad is shown. Unfortunately, even in this
restricted setting, estimation errors may break the truthfulness of VCG prices. We give an example in
Appendix~\ref{sec:a-ppc} showing that even if the auctioneer's estimates correctly identify the
welfare-maximizing allocation, they may not yield truthful prices, even in the special case where $\beta_i=1$.

In the language of allocations and payments, truthfulness is broken because the
auctioneer only knows an estimate of $\al$ and thus does not have enough information
to compute true VCG prices. However, once ads are shown, clicks may be measured, giving an
unbiased estimate of bidders' values. Unfortunately,
this can only be done once --- since the auctioneer only has one opportunity to show ads to the user,
these unbiased estimates can only be measured under a single advertiser-slot assignment. Fortunately, {\em these unbiased
estimates are exactly the information required to compute truthful payments using a single-call mechanism.}

Since a player's bid $b_i$ is merely its value-per-click $v_i$, this version of a PPC ad auction is a
single-parameter domain and we can apply the result of~\cite{BKS10}. Their result says that we
can turn any monotone allocation rule into a truthful-in-expectation mechanism --- maximizing welfare
subject to estimates $\alpha_j'$ and $\beta_i'$ is a monotone allocation rule as long as the estimates
$\alpha_j'$ have the same order as $\alpha_j$ (i.e. $\alpha_{j_1}'\geq\alpha_{j_2}'$ if $\alpha_{j_1}\geq\alpha_{j_2}$), so~\cite{BKS10} gives a truthful mechanism for almost any estimates:
\begin{theorem} 
Consider a single-parameter PPC auction with separable CTRs and let $\al^{PPC}$ be the allocation rule
that maximizes welfare using estimated CTR parameters $\alpha_j'$ and $\beta_i'$, where the estimates $\alpha_j'$ are properly ordered. Then $\BKSRedn(\al^{PPC},\gamma)$,
the single-call reduction of~\cite{BKS10}, gives a mechanism that is truthful in expectation and has expected
welfare within a factor of $(1-\gamma)^n$ of $\al^{PPC}$.
\end{theorem}

\paragraph{Outcome-Dependent Values and CTRs} While most research uses single-parameter
models for analytical tractability, an advertiser's value-per-click
$v$ really depends on the advertiser-slot assignment chosen by the auctioneer as noted earlier.
As in the preceding single-parameter setting, estimated CTRs are insufficient to guarantee
truthfulness; however, the reduction of~\cite{BKS10} no-longer applies in such a multi-parameter domain --- we show
how our MIDR single-call reduction can be used to recover truthfulness.

To capture the dependence on the advertiser-slot assignment, we assume that a bidder's CTR $c_{i,j}$ and value-per-click $v_{i,j}$ depend
arbitrarily on both the bidder $i$ and the slot $j$. Since the only allocation rules that have truthful prices in general multi-parameter domains
are MIDR, we assume that the auctioneer can generate a MIDR allocation, specifically
we assume the auctioneer can query an oracle to determine the allocation that maximizes the welfare of any set of bidders under the actual bid $b$ (but not necessarily for an arbitrary bid $b$) and apply our MIDR reduction:
\begin{theorem}\label{thm:adwordsApp} 
Consider a multi-parameter PPC auction where a bidder's value-per-click $v_{i,j}$ depends on the bidder and the slot.
Let $\al^{PPC}$ be an allocation rule that chooses the advertiser-slot assignment returned by the welfare-maximizing oracle described above. Then the mechanism $\optMidrRedn(\al^{PPC},\gamma)$
is truthful in expectation and approximates
the welfare of $\al^{PPC}$ to within a factor of $(1-\gamma)$.
\end{theorem}

\section{Single-parameter reductions}

\label{sec:sp-char}
In this section, we characterize truthful reductions for single-parameter
domains and show that the construction of~\cite{BKS10} is optimal. Theorem~\ref{thm:sp-char} characterizes all
reductions that are truthful for an arbitrary monotone, bounded,
single-parameter allocation function $\al$. Our characterization is more general
than the self resampling procedures
described by Babaioff et al. and shows that a wide variety of
probability measures may be used to construct a truthful reduction.
Theorem~\ref{thm:babaioff-opt} shows that the construction given in Babaioff et
al. is optimal among such reductions for a fixed bound on the precision, welfare
approximation, or revenue approximation of the reduction.

As in the MIDR setting, truthful payments give intuition for the structure of a
single-call reduction. As noted in Section~\ref{sec:prelim}, payments are
truthful if and only if they are given by the Archer-Tardos characterization:
\begin{equation}\label{eqn:sp_truthful_payment}
p_i(b) = b_i\al_i(b)-\int_{0}^{b_i}{\al_i(u,b_{-i})du}\enspace.
\end{equation}
Loosely speaking, this says ``charge $i$ the value she receives
minus what she would expect if she lowered her bid.'' Thus, a single call reduction should, with some probability, 
lower agents' bids to compute the value of allocation function at $(u,b_{-i})$
for $u \leq b_i$.

\subsection{Characterizing Single-Call Reductions} For the sake of intuition, we start with the
special case that the resampling measure $\mub$ has a nicely behaved density representation
$f_b(\hat b)$ (the resampling density) that is continuous in $\hat b$ and $b$. The
proof for arbitrary measures $\mu_b$ requires significant measure theory and is deferred until Appendix~\ref{sec:a-sp-char}. 

Define the coefficients $c^f_i(\hat b,b)$ as
$c^f_i(\hat b,b)=1-\frac{1}{b_i}\int_{0}^{b_i}\frac{f_{u,b_{-i}}(\hat
b)}{f_b(\hat b)}du$ when $b_i\neq0$, and to be $0$ when $b_i=0$.
We characterize truthful reductions as follows:
\begin{theorem} \label{thm:sp-char} 
A normalized single-parameter reduction $(f,\{\lambda_i\})$ for the set of all monotone
 bounded
single-parameter allocation functions satisfies truthfulness, individual
rationality and no positive transfers in an ex-post sense if and only if the following conditions are met: 
\begin{enumerate} 
\item The resampling density $f_b$ is such that the single-call mechanism's randomized allocation procedure $\asci(b)$ is monotone 
in expectation, i.e., for all agents $i$, for all $b$, and $b_i' \geq b_i$,
$\ex_{\hat b\sim f_b}[\asci(b'_i,b_{-i})]\geq \ex_{\hat b\sim f_b}[\asci(b)]$. (See below.)
\item The resampling density $f_b$ is such that $f_b(\hat b)\neq 0$ if
$\int_0^{b_i}f_{u,b_{-i}}(\hat b)du\neq 0$.\footnote{This condition
effectively requires $c^f_i(\hat b, b)$ to be finite.}
\item The payment functions $\lambda_i(A(\hat b),\hat b,b)$ satisfy: 
$\lambda_i(A(\hat b),\hat b,b)=b_ic_i^f(\hat b,b)A_i(\hat b)$ almost surely,
i.e. for all $\hat b$ except possibly a set with probability zero under $f_b$.
\end{enumerate} 
\end{theorem}

\begin{proof} (See Appendix~\ref{sec:a-sp-char} for the proof when $\mub$ is an arbitrary measure.)

{\bf Necessity.} 
The first condition, that $\asc$ must be monotone in expectation, follows directly from
Archer-Tardos characterization of truthful allocation functions. The second and
third conditions, as we prove below, are necessary for the expected payment to
take the form required by the Archer-Tardos characterization.

The allocation function $\asc$ is a single-parameter
allocation function, so the Archer-Tardos characterization gives truthful
prices if they exist:
\begin{align*}
\ex[\psci]&=b_i\ex_{\hat b\sim f_{b}}[\asci(b)]-\int_0^{b_i}\ex_{\hat
b\sim f_{u,b_{-i}}}[\asci(u,b_{-i})]du\\
&=b_i\ex_{\hat b\sim f_b}[\al_i(\hat b)]-\int_0^{b_i}\ex_{\hat b\sim f_{u,b_{-i}}}[\al_i(\hat b)]du\\
&=b_i\int_{\hat b\in\Re^n}\al_i(\hat b)f_b(\hat b)d\hat b-\int_0^{b_i}\int_{\hat
b\in\Re^n}\al_i(\hat b)f_{u,b_{-i}}(\hat b)d\hat bdu\enspace.
\end{align*}
Rearranging, where changing the order of integration may be justified by Tonelli's theorem, gives
\[\ex[\psci]=\int_{\hat b\in\Re^n}f_b(\hat b)b_i\al_i(\hat
b)\left(1-\frac{1}{b_i}\int_{0}^{b_i}\frac{f_{u,b_{-i}}(\hat b)}{f_b(\hat
b)}du\right)d\hat b\enspace.\] 
By construction, we can express the expected
price as 
\[\ex[\psci]=\int_{\hat b\in\Re^n}f_b(\hat b)\lambda_i(\al(\hat b),\hat
b,b)d\hat b\enspace.\] 
Thus truthfulness in expectation necessarily implies 
\begin{equation}\label{eqn:sp-char-eq}
\int_{\hat b\in\Re^n}f_b(\hat b)\lambda_i(\al(\hat b),\hat b,b)d\hat b=\int_{\hat
b\in\Re^n}f_b(\hat b)b_i\al_i(\hat
b)\left(1-\frac{1}{b_i}\int_{0}^{b_i}\frac{f_{u,b_{-i}}(\hat b)}{f_b(\hat
b)}du\right)d\hat b\enspace.
\end{equation}
Note that proving the necessity of condition three in the theorem is equivalent to proving that
the integrands in the LHS and the RHS of~\eqref{eqn:sp-char-eq} are equal
almost everywhere. That is, we have to show that the only way for Equation~\eqref{eqn:sp-char-eq}
to hold for all monotone bounded $\al$ is when the integrands are equal almost everywhere. To
show this, it is sufficient to show that Equation~\eqref{eqn:sp-char-eq} must
still hold if we restrict the range of integration to an arbitrary rectangular
parallelepiped (hence forth called as rectangle) $\rect\subseteq\Re^n$ (see why this is enough in Appendix~\ref{sec:a-sp-char} for a
more general setting), that is, it is sufficient to show that for all rectangles
$\rect\subseteq\Re^n$
\begin{equation}\label{eqn:sp-char-eqr}
\int_{\hat b\in \rect}f_b(\hat b)\lambda_i(\al(\hat b),\hat b,b)d\hat b=\int_{\hat
b\in \rect}f_b(\hat b)b_i\al_i(\hat
b)\left(1-\frac{1}{b_i}\int_{0}^{b_i}\frac{f_{u,b_{-i}}(\hat b)}{f_b(\hat
b)}du\right)d\hat b\enspace.
\end{equation}
Showing~\eqref{eqn:sp-char-eqr} would be straight-forward if we are given
that~\eqref{eqn:sp-char-eq} holds for {\em all} $\al$ --- we could take any
$\al$ and make it zero for all points not in $\rect$, and
then~\eqref{eqn:sp-char-eq} immediately implies~\eqref{eqn:sp-char-eqr}. 
However~\eqref{eqn:sp-char-eq} is guaranteed to be true only for monotone bounded $\al$,
since those are the allocation functions that could possibly be input to our
reduction. To see that it is still true when~\eqref{eqn:sp-char-eq} is
only guaranteed for monotone bounded $\al$, define the function $1_{\rect}(\hat b)$ as
\[1_{\rect}(\hat b)=\begin{cases}1,&\hat b\in \rect\\0,& otherwise.\end{cases}\]
Observe that $1_{\rect}$ can be written as $1_{\rect}(\hat b)=1_{\rect}^+(\hat
b)-1_{\rect}^-(\hat b)$
where $1_{\rect}^+$ and $1_{\rect}^-$ are both $\{0,1\}$, monotone functions. Moreover,
the functions $\al^+(b)=1_{\rect}^+(b)\al(b)$ and $\al^-(b)=1_{\rect}^-(b)\al(b)$ are also monotone, and they
agree with $\al$ on $\rect$. If we plug $\al^+$ and $\al^-$ into (\ref{eqn:sp-char-eq}) and
subtract the results, we get precisely (\ref{eqn:sp-char-eqr}). Thus condition three is necessary.

For the necessity of condition two, note that if it were not to hold, the coefficients $c^f_i$ will become
$-\infty$, and hence the payments as defined in condition three will not be
finite. Clearly finiteness of payments is a requirement. 

This proves that all three conditions in the theorem are necessary for truthfulness.

{\bf Sufficiency.}  We now show that the three stated conditions are
sufficient. In a single-parameter setting, for a mechanism to be truthful,
we need the allocation function to be monotone in expectation and the
payment function to satisfy the Archer-Tardos payment functions. 
Condition one guarantees that the allocation function output by the single-call
reduction is a monotone in expectation allocation function. It remains to show
that the second and third conditions result in payments that agree with
Archer-Tardos payments. Given condition two, finiteness of payments as defined
in condition three is satisfied. All we
need to show is that under the formula of $\lambda_i(A(\hat b), \hat b, b))$
described in condition three, the single-call payments match in expectation with
Archer-Tardos payments, i.e.,~\eqref{eqn:sp-char-eq} holds.
Since $c^f_i(\hat b,b)=1-\frac{1}{b_i}\int_{0}^{b_i}\frac{f_{u,b_{-i}}(\hat
b)}{f_b(\hat b)}du$, taking
\[\lambda_i(A(\hat b),\hat b,b)=b_ic_i^f(\hat b,b)A_i(\hat b)\quad a.s.\]
trivially satisfies (\ref{eqn:sp-char-eq}), implying that the reduction is truthful.
\end{proof}

Unfortunately, our assumption that $\mub$ has a density representation is
unreasonable. Most significantly, one would expect $\hat b=b$ with some nonzero
probability, implying that $\mub$ would have at least one atom for most
interesting distributions. In particular, the distribution used in the BKS
transformation has such an atom, so it cannot be analyzed in this fashion.

To handle general measures $\mub$ we apply the same ideas using tools
from measure theory. A full proof is given in Appendix~\ref{sec:a-sp-char}.

\subsection{The BKS Reduction for Positive Types} The central construction
of Babaioff, Kleinberg, and Slivkins~\shortcite{BKS10} is a reduction for scenarios where bidders have positive types.\footnote{They also give a reduction that applies to more general type spaces, but we do not state it here.}

Their resampling procedure (implicitly defining $\mub$) is described Algorithm~\ref{alg:bkssc}.
In the language of our characterization, the coefficients $c_i^\fbks$ are
\begin{equation*}
c_i^\fbks(\hat b,b)=\begin{cases}
1,&\hat b_i=b_i\\
1-\frac1{\gamma}&otherwise.
\end{cases}
\end{equation*}
They proved that $\BKSRedn(\al,\gamma)$ is truthful. This fact can be easily
derived from Theorem~\ref{thm:sp-char}:
\begin{theorem}[Babaioff, Kleinberg, and Slivkins 2010.] For all monotone, bounded, single-parameter allocation rules $\al$, the single-call mechanism given by
$\BKSRedn(\al,\gamma)$ satisfies truthfulness and no positive transfers in an ex-post sense and is ex-post universally individually rational.
\end{theorem}

\begin{algorithm}[tb]\label{alg:bkssc}
\SetKwInOut{Input}{input}\SetKwInOut{Output}{output}
\SetKwIF{WP}{ElseWP}{Otherwise}{with probability}{}{with probability}{otherwise}{end}
\Input{Bounded, monotone allocation function $\al$.}
\Output{Truthful-in-expectation mechanism $\me=(\asc,\{\psci\})$.}
\BlankLine
\nl Solicit bids $b$ from agents\;
\nl \For{$i\in[n]$}{
 \eWP {$1-\gamma$}{
   Set $\hat b_i=b_i$\;}{
   Sample $x_i$ uniformly at random from $[0,\hat b_i]$\;
   Set $\hat b_i=b_i x_i^{\frac1{1-\gamma}}$\;}}
\nl Realize the outcome $\al(\hat b)$\;
\nl Charge payments\\
$\lambda_i(\al(\hat b^M),\hat b^M,b)= b_i\al_i(\hat b)\times\begin{cases}1,&\hat b_i=b_i\\\frac{1-\gamma}{\gamma},&\hat b_i<b_i\end{cases}$\;
\caption{$\BKSRedn(\al,\gamma)$ --- The BKS reduction for single-parameter domains}
\end{algorithm}

\subsection{Optimal Single-Call Reductions}
\label{sec:babaioff}

\label{sec:sp-char-opt}

Analogous to our MIDR construction, we show that, the BKS construction for positive types
is optimal with respect to precision, welfare, and revenue as defined in Section~\ref{sec:sc} (other type spaces are discussed in Appendix
\ref{sec:a-sp-char}). Using our characterization from Theorem~\ref{thm:sp-char},
the bid-normalized payments we wish to optimize will be
\[\sum_j \frac{\lambda_{ij}(b_j(\al(\hat
b)),\hat b,b)}{b_j(\al(\hat b))}=\frac{c_i^\mu(\hat
b,b)b_i\al_i(\hat b)}{b_i\al_i(\hat b)}=c_i^\mu(\hat
b,b)\enspace.\]
Thus, optimizing variance of normalized payments is equivalent
to optimizing $\max_{i}\var_{\hat b\sim\mub}c_i^\mu(\hat b,b)$,
and optimizing the worst-case normalized payment is equivalent to
optimizing $\sup_{i,\hat b}|c_i^\mu(\hat b,b)|$.

For this section, we make a ``nice
distribution'' assumption that for any $u\neq b_i$, $\Pr(\hat b_i=u|b)=0$.
That is, if we compute the marginal distribution of $\hat b_i$, the only bid $\hat b_i$ that has an atom is $b_i$
(other bids only have positive density). We handle the
general case in the full proofs in Appendix~\ref{sec:a-babaioff}.

Our main result is that the BKS transformation is optimal:
\begin{theorem}
\label{thm:babaioff-opt}
The single-call reduction $\BKSRedn(\al,\gamma)$ optimizes the variance of bid-normalized payments and the worst-case bid-normalized payment

for every $b$ subject to a lower bound $\alpha=(1-\gamma)^n\in(\frac1e,1)$ on the precision, the welfare approximation, or the revenue approximation.
\end{theorem}

To prove Theorem~\ref{thm:babaioff-opt}, we first show that the three metrics we study are equivalent for interesting reductions in the single parameter setting:
\begin{lemma}\label{lem:sp-swrapx}
For $\alpha>\frac1e$ and $n\geq2$, a reduction that optimizes the variance of normalized payments or the maximum normalized payment subject to a precision constraint of $\Pr(\hat b= b|b)\geq\alpha$ also optimizes the maximum payment subject to a welfare or revenue approximation of $\alpha$.
\end{lemma}
\begin{proof} (Sketch. The full proof is in Appendix~\ref{sec:a-babaioff}.) Consider the following allocation function:
\[\al_i(b)=\begin{cases}
1,&b\geq\bar b\\
0,&otherwise.
\end{cases}\]
Intuitively, a reduction should not resample to higher bids because
Archer-Tardos payments do not depend on higher bids, and hence no useful
information is obtained through raising bids. However, if a reduction never raises bids 
(i.e. $\Pr(\hat b\leq b|b)=1$), then the welfare and revenue of a single-call
reduction will both be precisely $\pr(\hat b=b|b)$ if we consider the above
mentioned $\al$ at a bid of $\bar b$. 
\end{proof}

Thus, to prove Theorem~\ref{thm:babaioff-opt}, it is sufficient to prove that the BKS reduction optimizes precision.
\begin{theorem}
\label{thm:babaioff-relnr}
The single-call reduction $\BKSRedn(\al,\gamma)$ optimizes the variance of normalized payments and the worst-case normalized payment among reductions with a precision of at least $\alphap=(1-\gamma)^n>\frac1e$.
\end{theorem}

\begin{proof} {\em (Sketch. The full proof is in Appendix~\ref{sec:a-babaioff}.)} When $\Pr(\hat b=b|b)$ is large, the mechanism extracts a modest payment from $i$ when $\hat b_i=b_i$ and pays a large rebate otherwise. Thus, we bound $\inf_{\hat b,i} c_i^\mu(\hat b,b)$. Let $\pi^\mu(M,b)$ be the probability (given $b$) that $\hat b_i=b_i$ for all $i\in M$ and $\hat b_i<b_i$ for all $i\not\in M$. Then the key step is to prove the following lower bound on $\inf c_i^\mu$:
\begin{equation*}
\inf_{\hat b} c_i^\mu(\hat b, b)\leq-\frac{\pi^\mu(M\cup\{i\},b)}{\pi^\mu(M,b)}\enspace.
\end{equation*}
Notably, this bound takes the same form as the truthful payment coefficients for MIDR reductions. Applying the same logic as Theorem \ref{thm:vcg-opt} shows that the BKS transformation is optimal.
\end{proof}

\section{Acknowledgments} First, we would like to thank Kamal Jain for enlightening us about the problems that can arise in pay-per-click advertising auctions and for suggesting~\cite{BKS10} as a possible solution.

Second, we would like to thank many people who have provided invaluable feedback and suggestions. In particular, we would like to thank Alex Slivkins for his substantial help in the revising process, Robert Kleinberg for his technical suggestions, and Christos Papadimitriou for his helpful comments.

\bibliographystyle{alpha}
\bibliography{sc}

\appendix

\section{A PPC Auction Example}
\label{sec:a-ppc}
The following example
illustrates how the welfare optimal assignment may be robust to inaccuracies in the
CTR estimates $c'$ but the truthful payments are quite fragile.

\begin{example}\label{ex:ppc-sp}
Consider a 2-slot, 2-advertiser setting with CTRs $c_j$ and bids $b_i$. Assume that $b_1>b_2$ and $c_1>c_2$, so that the welfare-optimizing assignment
 is to assign ad-1 to slot-1 and ad-2 to slot-2, i.e.,
 \begin{equation}\label{eqn:actWelfare}
 c_{1}b_{1} + c_{2}b_{2} \geq c_{1}b_{2} + c_{2}b_{1}\enspace.
\end{equation}
The auctioneer wishes to optimize welfare, so he uses $c_j'$ to implement the VCG allocation.
It is quite plausible that maximizing welfare w.r.t $c_{j}'$ results in the same
welfare maximizing allocation, namely given~\eqref{eqn:actWelfare}, it is not unreasonable to assume
that the following is true if the auctioneer's estimates are good enough:
\begin{equation*}\label{eqn:oraclePPC}
c_{1}'b_{1} + c_{2}'b_{2} \geq c_{1}'b_{2} + c_{2}'b_{1}\enspace.
\end{equation*}
However, we will show that this is not enough to guarantee truthfulness.

We show that advertiser-1 may have an incentive to lie. According to the estimates $c_j'$,
The expected VCG payment should be $c_1'b_2-c_2'b_2$. Since advertiser 1 will only be charged
when he actually receives a click, the price-per-click charged will be
\[p_1=\frac1{c_1'}[c_1'b_2-c_2'b_2]\enspace.\]
and the expected utility to bidder $i$ will be
\[u_1=c_1\left(b_1-\frac1{c_1'}[c_1'b_2-c_2'b_2]\right)\enspace,\]
where the extra $c_{1}$ gets multiplied because the utility is non-zero
only upon a click, which happens with probability $c_{1}$.

Now, for example, let the inaccurate $c_{j}'$ be as follows: $c_{1}' = \alpha c_{1}$, $c_{2}'
= c_{2}$ where $\alpha>1$. Notice that in this example we always have 
 \begin{equation*}
 \alpha c_{1}b_{1} + c_{2}b_{2} \geq \alpha c_{1}b_{2} + c_{2}b_{1}
\end{equation*}
and thus the mechanism will always maximize welfare in spite of the estimation errors.

The utility of advertiser-1 will be
$$u_1=c_{1}\left(b_{1} - \frac{1}{\alpha c_{1}}[\alpha c_{1}b_{2} -
c_{2}b_{2}]\right).$$
Now, suppose
advertiser-1 decides to lie and bid zero, he gets the second slot, pays zero,
and gets utility of $c_{2}b_{1}$. Lying is clearly profitable if 
$$c_{2}b_{1} > c_{1}\left(b_{1} - \frac{1}{\alpha c_{1}}[\alpha c_{1}b_{2}
- c_{2}b_{2}]\right).$$
Rearranging, lying is profitable if
\begin{equation}\label{eqn:lyingPPC}
\alpha c_{1}b_{1} + c_{2}b_{2} <  \alpha
c_{1}b_{2} + \alpha c_{2}b_{1}
\end{equation}
It is quite possible that 
lying might be profitable, that is inequality~\eqref{eqn:lyingPPC} holds true.
For example, if $c_1=0.1$, $c_2=0.09$, $b_1=1.1$, $b_2=1$, and $\alpha=1.1$,
payments computed using $c_j'$ are nontruthful, even though the mechanism always picks
the welfare-maximizing assignment for any $\alpha>1$.
\end{example}

\section{Optimality Proofs for MIDR Reductions}
\label{sec:a-midr-opt}

\subsection{Optimizing Social Welfare}
\begin{lemma}[Restatement of Lemma~\ref{lem:vcg-swapx}]\label{lem:welfareApprox}
The reduction $\optMidrRedn(\al,\gamma)$ obtains an
$\alpha_\pi=\min_i\Pr_\pi(i\in M)$ approximation to the social welfare, and
there is an allocation function $\al$ and bid $b$ such that this bound is
tight. 
\end{lemma}

\begin{proof}
The expected social welfare of the single-call mechanism, where the expectation is over
the randomness in the resampling function is given by
$\ex\left[\sum_{j\in[n]}b_j(\asc(b))\right]$. 
We now prove the required lower bound on this quantity. 

\begin{align*}
\ex\left[\sum_{j\in[n]}b_j(\asc(b))\right]
=\sum_{j\in[n]}\sum_{M\subseteq[n]}\pi(M)b_j(\al(\hat b^M))
&=\sum_{M\subseteq[n]}\pi(M)\sum_{j\in[n]}b_j(\al(\hat b^M))\\
&\geq\sum_{M\subseteq[n]}\pi(M)\sum_{j\in M}b_j(\al(\hat b^M))\\
&\geq\sum_{M\subseteq[n]}\pi(M)\sum_{j\in M}b_j(\al(\hat b^{[n]}))\\
&=\sum_{j\in[n]}\Pr_\pi(j\in M)b_j(\al(\hat b^{[n]}))\\
&\geq\left(\min_{j\in[n]}\Pr_\pi(j\in M)\right)\sum_{j\in[n]}b_j(\al(\hat
b^{[n]}))\\
&=\left(1-\max_{j\in[n]}\Pr_\pi(j\not\in M)\right)\sum_{j\in[n]}b_j(\al(\hat
b^{[n]}))\enspace.
\end{align*}

Finally, we observe that this is tight. Consider a valuation and allocation
function pair for which, every agent other than some agent $j$ has a zero value
for every outcome,
and agent $j$ has a non-zero value only for those outcomes that were chosen
taking $j$ into consideration, i.e., :
$$
b_k(\al(\hat b^M))=\begin{cases}
0,&k\neq j\\
0,&j\notin M\\
1,&\text{otherwise}
\end{cases}
$$
When $j=\argmax_{k\in[n]}\Pr_\pi(k\notin M)$, the preceding bound is tight.
\end{proof}

\begin{lemma}\label{lem:vcg-sw-maxbound}
Let $\pi$ be a distribution such that
$\max_{i,M}c_i^{\pi}(M)<\max_{i,M}c_i^{\bar\pi}(M)$. Then
$$\max_i\Pr_{\pi}(i\notin M) > \max_i\Pr_{\bar\pi}(i\notin M)\enspace.$$
\end{lemma}

\begin{proof} 

Let $\bar c=\max_{i,M}c_i^{\bar\pi}(M)$. Note that for all $M|i\not\in M$, 
$c_i^{\bar\pi}(M)=\frac{\pibar(M\cup\{i\})}{\pibar(M)}=\bar c$. It follows by algebra that
$\frac{\sum_{M|i\notin M} \bar\pi(M\cup\{i\})}{\sum_{M|i\notin M}\bar\pi(M)}=\bar c$
and therefore by the conditions of the lemma

\begin{equation}\label{eqn:welfareLemmaEq1}
\max_{i,M}c_i^{\pi}(M)<\frac{\sum_{M|i\notin M}
\bar\pi(M\cup\{i\})}{\sum_{M|i\notin M}\bar\pi(M)}\enspace.
\end{equation}
Next we have,
\begin{equation}\label{eqn:welfareLemmaEq2}
\max_{i,M}c_i^\pi(M) \geq \max_{M}c_i^\pi(M) \geq \max_{M|i\notin M}
\frac{\pi(M\cup\{i\})}{\pi(M)} \geq \frac{\sum_{M|i\notin M}
\pi(M\cup\{i\})}{\sum_{M|i\notin M}\pi(M)}
\end{equation}
Combining~\eqref{eqn:welfareLemmaEq1} and~\eqref{eqn:welfareLemmaEq2} gives 
\begin{equation}\label{eqn:welfareLemmaEq3}
\frac{\sum_{M|i\notin M} \pi(M\cup\{i\})}{\sum_{M|i\notin M}\pi(M)}
<\frac{\sum_{M|i\notin M} \bar\pi(M\cup\{i\})}{\sum_{M|i\notin M}\bar\pi(M)}
\end{equation}
Note that since $\pi$ and $\pibar$ are probability distributions, the sum of the
numerator and denominator of both the LHS and the RHS
of~\eqref{eqn:welfareLemmaEq3} equals 1. Thus, it immediately follows that the
denominator of the LHS is larger than the denominator of the RHS, i.e.,
\begin{equation}\label{eqn:welfareLemmaEq4}
\sum_{M|i\notin M}\pi(M) > \sum_{M|i\notin M}\bar\pi(M)
\end{equation}
Inequality~\eqref{eqn:welfareLemmaEq4} when restated, reads as 
$$\Pr_{\pi}(i\notin M) > \Pr_{\pibar}(i\notin M)\enspace.$$
But since the above inequality is true for all $i$, and the RHS of
the above inequality is the same for all $i$ (namely the parameter $\mu$ by
which the reduction is parametrized), the statement of the lemma follows. 
\end{proof}

\begin{theorem}(Restatement of Theorem~\ref{thm:vcg-swopt}.)
The reduction $\optMidrRedn(\al,\gamma)$ minimizes payment variance
and the worst-case payment among all reductions that achieve a welfare approximation of at least $\alphaw=1-\gamma$.
\end{theorem}
\begin{proof} By Lemma \ref{lem:welfareApprox}, the worst case loss in social welfare of a distribution $\pi$ is given by
\[ 1 - \alpha_\pi=\max_i\Pr_\pi(i\not\in M)\enspace.\]
For worst-case payments, the contrapositive of Lemma \ref{lem:vcg-sw-maxbound} precisely says that if
$1-\alpha_\pi\leq 1-\alpha_{\pibar}$, then the largest payment
$\max_{M,i}c_i^\pi(M)\geq\max_{M,i}c_i^{\bar\pi}(M)$, thus proving that any other reduction will be worse.

For payment variance, arguing along the lines of Theorem~\ref{thm:vcg-char} again says that variance will be minimized
when $\pi$ is an independent distribution and $\pr(i\in M)$ is the same for all $i$. Since $\pibar$ is precisely the distribution
that does this, it follows that it is optimal.
\end{proof}

\subsection{Optimizing Revenue}

\begin{lemma}[Restatement of Lemma~\ref{lem:vcg-rapx}]
The reduction $\optMidrRedn(\al,\gamma)$ obtains an $\alpha_\pi=\pi([n])$ approximation to the revenue, and this is tight.
\end{lemma}
\begin{proof}
For any $b$ with non-negative valuations, the revenue under a single call reduction will be
\begin{eqnarray*}
\sum_{i\in[n]}\ex[\psci]\TCSeas&=&\sum_{i\in[n]}\sum_{M\subseteq[n]}\pi(M)\left(\sum_{k\neq
i} b_k(\al(\hat b^{M\setminus\{i\}}))-\sum_{k\neq i} b_k(\al(\hat b^{M}))\right)\\
&\geq&\pi([n])\sum_{i\in[n]}\left(\sum_{k\neq i} b_k(\al(\hat
b^{[n]\setminus\{i\}}))-\sum_{k\neq i} b_k(\al(\hat b^{[n]}))\right)
\end{eqnarray*}
where $\sum_{i\in[n]}\left(\sum_{k\neq i} b_k(\al(\hat
b^{[n]\setminus\{i\}}))-\sum_{k\neq i} b_k(\al(\hat b^{[n]}))\right)$ is the revenue generated by 
$\al$ under VCG prices. Thus, any distribution $\pi(M)$ gives an $\alpha=\pi([n])$ approximation to the revenue.

To see that this is tight, consider the following allocation function:
$$b_i(\al(\hat b^M))=\begin{cases}
\frac{1}{n}&\mbox{$M = [n]$}\\
\frac{1}{n-1}&i\in M\mbox{ but }M \neq [n]\\
0,&\text{otherwise.}
\end{cases}$$

The revenue under VCG prices is  $\sum_{i\in[n]}\left(\sum_{k\neq i} b_k(\al(\hat
b^{[n]\setminus\{i\}}))-\sum_{k\neq i} b_k(\al(\hat b^{[n]}))\right)$, which is
$n(\frac{n-1}{n-1}-\frac{n-1}{n}) = 1$. 

Under any single-call reduction, the revenue will be given by
\begin{eqnarray*}
\sum_{i\in[n]}\ex[\psci]\TCSeas&=&\sum_{i\in[n]}\sum_{M\subseteq[n]}\pi(M)\left(\sum_{k\neq
i} b_k(\al(\hat b^{M\setminus\{i\}}))-\sum_{k\neq i} b_k(\al(\hat b^{M}))\right)\\
&=&\sum_{i\in[n]}\pi([n])\left(\sum_{k\neq i} b_k(\al(\hat
b^{[n]\setminus\{i\}}))-\sum_{k\neq i} b_k(\al(\hat b^{[n]}))\right)\\
&=&\sum_{i\in[n]}\pi([n])\left(1-\frac{n-1}{n}\right)\\
&=&\pi([n])\enspace.
\end{eqnarray*}
\end{proof}

\section{Characterizing Reductions for Single-Parameter Domains}
\label{sec:a-sp-char}

In this section we characterize truthful single-call reductions for single-parameter domains that
use arbitrary measures $\mub$. We refer the reader to Section~\ref{sec:a-analysis} for some
background and definitions from measure theory.

Before we begin, we must formalize some properties of the functions $\al$ and the measures $\mub$.
The following assumptions would typically be implicit in Algorithmic Mechanism Design;
however, it is necessary that they be formalized for some of the tools in our proof. We assume the following: 
\begin{enumerate}
\item Any allocation function $\al$ that the reduction receives as
input (as a black box) is a Borel measurable function, i.e., each of the $\al_i$'s
as a function from $\mathbb{R}^n \rightarrow \mathbb{R}_+$ is a bounded Borel measurable
function.
\item For every $b$, the resampling measure $\mub(\cdot)$ is 
a Borel probability measure.
\item The function mapping the bid $b$ to the resampling measure
$\mub(\cdot)$ is measurable w.r.t to the Borel $\sigma$-algebra on the space of
Borel probability measures over $\Re^n$.
\end{enumerate}

First, we use the measure $\mub(\cdot)$ to define a signed measure
$\mubpsci(B) = b_i\mub(B)-\int_{0}^{b_i}\mub[u,b_{-i}](B)du$
which has the property:
\[\int_{\hat b\in\Re^n}\al_i(\hat b)d\mubpsci=b_i\ex_{\hat b\sim \mub}[\al_i(\hat
b)]-\int_0^{b_i}\ex_{\hat b\sim \mu_{u,b_{-i}}}[\al_i(\hat b)]du\enspace,\]
that is, integrating $\al_i$ with respect to $\mubpsci$ is equivalent to computing
the Archer-Tardos prices.

\begin{lemma}\label{lem:mubpsci}
The function $\mubpsci(B) = b_i\mub(B)-\int_{0}^{b_i}\mub[u,b_{-i}](B)du$ is a finite signed measure satisfying
\[\int_{\hat b\in\Re^n}\al_i(\hat b)d\mubpsci=b_i\ex_{\hat b\sim \mub}[\al_i(\hat
b)]-\int_0^{b_i}\ex_{\hat b\sim \mu_{u,b_{-i}}}[\al_i(\hat b)]du\]
for any bounded $\al_i$
\end{lemma}
\begin{proof}
First, we show that $\mubpsci(B)$ is a finite signed measure. Since $\mub$ is a probability measure, we have
$\mub(B)\leq1$ for all $B$. Thus, $\mubpsci(B)$ is well-defined and finite for all Borel sets $B$ (note that
the integral is well defined by our assumptions on the measurability of $\mub$). From this it is easy to see
that $\mubpsci(\emptyset)=0$ because $\mub(\emptyset)=0$. It remains to show countable additivity, i.e.
$\sum_{k=1}^\infty\mubpsci(B_k)=\mubpsci(\cup_k B_k)$, which follows because integrals obey countable
additivity for nonnegative functions (see Fact~\ref{fct:ana-integral-countadd}):
\begin{align*}
\sum_{k=1}^\infty\mubpsci(B_k)&=\sum_{k=1}^\infty\left(b_i\mub(B_k)-\int_{0}^{b_i}\mub[u,b_{-i}](B_k)du\right)=\sum_{k=1}^\infty b_i\mub(B_k)-\int_{0}^{b_i}\sum_{k=1}^\infty\mub[u,b_{-i}](B_k)du\\
&=b_i\mub(\cup_kB_k)-\int_{0}^{b_i}\mub[u,b_{-i}](\cup_kB_k)du=\mubpsci(\cup_kB_k)\enspace.
\end{align*}

Second, we show from first-principles that integrating $\al_i$ with respect to
$\mubpsci$ is equivalent to calculating the Archer-Tardos prices for $\al_i$.
We begin by showing this equality for charateristic functions over Borel
measurable sets. The proof for more general functions (in our case
$\al_i$) can be built-up from characteristic functions precisely as in the
definition of an integral, so we omit it (see Definition~\ref{def:integral}).
Let $1_B$ be the characteristic function of a Borel measurable set. By definition of an integral, $\int 1_Xd\nu=\nu(X)$,
and plugging in we observe the desired equality:
\begin{align*}
b_i\ex_{\hat b\sim \mub}[1_B(\hat
b)]-\int_0^{b_i}\ex_{\hat b\sim \mu_{u,b_{-i}}}[1_B(\hat b)]du&= b_i\mub(B)-\int_{0}^{b_i}\mub[u,b_{-i}](B)du\\
&=\int_{\hat b\in\Re^n}1_B(\hat b)d\mubpsci\enspace.
\end{align*}
\end{proof}

The general version of the characterization theorem shows that the payment functions
precisely correspond to the density function $\rhobi(\hat b)$ relating $\mubpsci$ to $\mub$ 
(i.e. the Radon-Nikodym derivative of $\mubpsci$ with respect to $\mub$ --- its
existence is guaranteed by the absolute continuity that 
figures in the characterization theorem~\ref{thm:a-sp-char} below). In this setting, we
can equivalently define the associated coefficients $c^\mu_i(\hat b,b)$ as the function that satisfies
\begin{equation*}
b_ic_i^\mu(\hat b,b)=\rhobi(\hat b)\enspace.
\end{equation*}

\begin{theorem}[Characterizing single-call reductions] \label{thm:a-sp-char}{\em (Generalization of Theorem \ref{thm:sp-char}.)}

A single-call single-parameter reduction $(\mu,\{\lambda_i\})$ for the set of
all monotone bounded single-parameter allocation functions satisfies
truthfulness, individual rationality, and no positive transfers in expectation
if and only if the following conditions are met: 
\begin{enumerate} 
\item The distribution $\mu$ is such that for all monotone, locally bounded
$\al$, the randomized allocation procedure $\asci(b)$ is monotone in
expectation, i.e., for all agents $i$, for all $b$, and $b_i' \geq b_i$,
$\ex[\asci(b)]\leq\ex[\asci(b',b_{-i})]$ (see Lemma~\ref{lem:a-sp-char-mono} for
further discussion).
\item For all $i$, and for all Borel measurable sets $B$, 
the measure $\mub(B) \neq 0$ if $\int_{0}^{b_i}\mub[u,b_{-i}](B)du \neq 0$, or
equivalently, the signed measure $\mubpsci$ is absolutely continuous w.r.t.
measure $\mub$.
\item The payment functions $\lambda_i(A(\hat b),\hat b,b)$ satisfy 
\[\lambda_i(\al(\hat b),\hat b,b)=\rhobi(\hat b)\al_i(\hat b)+\lambda_i^0(\hat b,b)\enspace a.s.\]
where $\ex_{\hat b\sim \mub}[\lambda_i^0(\hat b,b)]=0$ and $\rhobi(\hat b)$ is the
 density function relating $\mubpsci$ to $\mub$.

(Almost surely, or $a.s.$, means that it holds everywhere except for a set with
measure zero under $\mub(\cdot)$.)
\end{enumerate}
\end{theorem}

\begin{proof} 
\paragraph{Necessity}
We first prove the necessity of the three conditions above. The
first condition, that $\asc$ is monotone in expectation, follows directly from
Archer-Tardos characterization of truthful allocation functions. The second and
third conditions, as we prove below, are necessary for the expected payment to
take the form required by the Archer-Tardos characterization.

We now write down the truthful payments give by the Archer-Tardos
characterization, and rewrite it using the signed measure $\mubpsci$. 
\begin{align*}
\ex[\psci]\TCSeas&=b_i\ex[\asci(b)]-\int_{0}^{b_i}\ex[\asci(u,b_{-i})]du\\
&=b_i\ex_{\hat b\sim \mub}[\al_i(\hat b)]-\int_0^{b_i}\ex_{\hat b\sim \mu_{u,b_{-i}}}[\al_i(\hat b)]du\\
&=\int_{\hat b\in\Re^n}{A_i(\hat b)d\mubpsci}\enspace.
\end{align*}
where the last equality follows from the definition of the signed measure
$\mubpsci$, and Lemma~\ref{lem:mubpsci}.

By definition of the reduction, we can write the expected payment as:
\begin{equation*}
\ex[\psci]=\int_{\hat b\in\Re^n}\lambda_i(\al(\hat b),\hat b,b)d\mub\enspace.
\end{equation*}
Equating these two gives
\begin{equation}\int_{\hat b\in\Re^n}\lambda_i(\al(\hat b),\hat b,b)d\mub=\ex[\psci]\TCSe=\int_{\hat b\in\Re^n}{A_i(\hat b)d\mubpsci}\enspace.\label{eqn:sp-char-inteq}\end{equation}
Next, we define the normalized payment function $\tilde\lambda$ as
\[\tilde\lambda_i(\al(\hat b),\hat b,b)=\lambda_i(\al(\hat b),\hat b,b)-\lambda_i(0^n,\hat b,b)\enspace.\]
By (\ref{eqn:sp-char-inteq}), $\int_{\hat b\in\Re^n}\lambda_i(0^n,\hat b,b)d\mub(B)=0$, and therefore we may write
\[\int_{\hat b\in\Re^n}\tilde\lambda_i(\al(\hat b),\hat b,b)d\mub=\int_{\hat b\in\Re^n}{\al_i(\hat b)d\mubpsci}\enspace.\]

If the above equality were to hold for all bounded, monotone, measurable
allocation functions $A$, then by Lemma~\ref{lem:sp-char-restrictedint}, this
implies for all Borel measurable sets $X\subseteq\Re^n$:
\begin{equation}\label{eqn:sp-char-allx}\int_{\hat b\in X}\tilde\lambda_i(\al(\hat b),\hat b,b)d\mub=\int_{\hat b\in X}{\al_i(\hat b)d\mubpsci}\enspace.\end{equation}
This statement would be intuitive if we allowed $\al_i$ to be any function ---
we could pick the function $\al_i'(b)=1_X(b)\al_i(b)$, i.e. we could zero
$\al_i$ except on $X$, and plug back into the previous equality. Unfortunately,
this $\al_i'$ is not monotone. The work of Lemma~\ref{lem:sp-char-restrictedint}
is to show that the space of bounded, monotone functions is still sufficiently
general as to guarantee equality for any Borel measurable set $X$.

Having derived Equation~\eqref{eqn:sp-char-allx}, we now show how it makes conditions two and three
in theorem necessary.  If we substitute the constant function $\al_i(b)=1$ into
(\ref{eqn:sp-char-allx}), we see that for all measurable $X$ \[\int_{\hat b\in
X}\tilde\lambda_i(1^n,\hat b,b)d\mub=\int_{\hat b\in X}d\mubpsci\enspace,\] that
is, $\tilde\lambda_i(1^n,\hat b,b)$ satisfies the definition of the derivative
of $\mubpsci$ w.r.t $\mub$, and therefore $\rhobi(\hat b) =
\tilde\lambda_i(1^n,\hat b,b)$.  Thus, given that finite 
payments $\lambda$ exist it follows that the density relating $\mubpsci$ to $\mub$, namely $\rhobi(\hat
b)$, also exists and is finite. But given that both $\mub$ and $\mubpsci$ are finite
measures, this also means that $\mubpsci$ is absolutely continuous w.r.t.
$\mub$. If not, then there exists a Borel measurable set $V$ such
that $\mubpsci(V)\neq0$ but $\mub(V)=0$. We run into an
immediate contradiction as follows:
\[0=\int_{\hat b\in
V}\rhobi(\hat b)d\mub=\int_{\hat b\in V}d\mubpsci=\mubpsci(V)\neq 0.\]

Thus we have proved that condition two, absolute continuity of $\mubpsci$ w.r.t. $\mub$, is necessary.

Returning to (\ref{eqn:sp-char-allx}), by the definition of $\rhobi(\hat b)$ we can write
\[\int_{\hat b\in X}\tilde\lambda_i(\al(\hat b),\hat b,b)d\mub=\int_{\hat b\in X}{\al_i(\hat b)\rhobi(\hat b)d\mub}\]

\[\int_{\hat b\in X}\left(\tilde\lambda_i(\al(\hat b),\hat b,b)-\al_i(\hat b)\rhobi(\hat b)\right)d\mub=0\]
for all Borel measurable sets $X\subseteq\Re^n$. By a standard argument 
(Fact~\ref{fact:ana-fzeroae}), 
this implies
$$\tilde\lambda_i(\al(\hat b),\hat b,b)-\al_i(\hat b)\rhobi(\hat b)=0$$
almost surely with respect to $\mub(B)$, the third condition. Thus we have shown that all the three
conditions are necessary.

\paragraph{Sufficiency} We now show that the three stated conditions are
sufficient. In a single-parameter setting, for a mechanism to be truthful,
we simply need the allocation function to be monotone in expectation, and the
payment function must satisfy the Archer-Tardos payment functions. 
Condition one guarantees that the allocation function output by the single-call
reduction is a monotone in expectation allocation function. It remains to show
that the second and third conditions result in payments that agree with
Archer-Tardos payments. Given condition two, we see that $\mubpsci$ is
absolutely continuous w.r.t the resampling measure $\mub$, and thus by Radon
Nikodym theorem, the density function $\rhobi(.)$ is finite and exists. All we
need to show is that under the formula of $\lambda_i(A(\hat b), \hat b, b))$
described in condition three, we have $$\int_{\hat
b\in\Re^n}\lambda_i(\al(\hat b),\hat b,b)d\mub =
b_i\ex[\asci(b)]-\int_{0}^{b_i}\ex[\asci(u,b_{-i})]du.$$ Once we substitute the
formula for $\lambda_i(A(\hat b),\hat b, b)$ from condition three, this equality 
follows from the definition of $\rhobi(\cdot)$ and $\mubpsci$.
\end{proof}

\begin{lemma}\label{lem:sp-char-restrictedint} 
Let $\mu$ and $\nu$ be finite measures (possibly signed), and let
$g:\Re^n_+\times\Re^n\rightarrow\Re$ be a function with $g(0,\hat b)=0$
satisfying $$\int_{\hat b\in\Re^n}g(\al(\hat b),\hat b)d\mu\TCSdd=\int_{\hat
b\in\Re^n}{A_i(\hat b)d\nu}$$ for all Borel measurable functions
$\al:\Re^n\rightarrow\Re_+^n$ where $\al$ is bounded and monotone in the sense
that $b'\geq b\Rightarrow\al(b')\geq\al(b)$.

Then for any such $\al$ and all Borel measurable sets $X\subseteq\Re^n$,
$$\int_{\hat b\in X}g(\al(\hat b),\hat b)d\mu\TCSdd=\int_{\hat b\in X}{A_i(\hat b)d\nu}\enspace.$$
\end{lemma}
\begin{proof} First, assume that the characteristic function of $X$ can be written
as the difference of two $\{0,1\}$ monotone functions, that is, $1_X(b)=f^+(b)-f^-(b)$
where $f^+$ and $f^-$ are monotone functions mapping $\Re^n$ to $\{0,1\}$. Note that
this includes all rectangular parallelepipeds (a product of open, closed, or half-open
intervals).

Define as $\al_i^+(b)=\al_i(b)\cdot f^+(b)$ and $\al_i^-(b)=\al_i(b)\cdot f^-(b)$.
Note that for any bounded, monotone, measurable $\al$, the functions $\al^+$ and
$\al^-$ are similarly bounded and monotone. Therefore the conditions of the lemma imply
\[\int_{\hat b\in\Re^n}g(\al^+(\hat b),\hat
b)d\mu\TCSdd=\int_{\hat b\in\Re^n}A_i^+(\hat b)d\nu\]
and
\[\int_{\hat
b\in\Re^n}g(\al^-(\hat b),\hat b)d\mu\TCSdd=\int_{\hat b\in\Re^n}A_i^-(\hat
b)d\nu\]
Taking the difference, we get $$\int_{\hat b\in\Re^n}\left(g(\al^+(\hat
b),\hat b)-g(\al^-(\hat b),\hat b)\right)d\mu\TCSdd=\int_{\hat
b\in\Re^n}\left(A_i^+(\hat b)-\al_i^-(\hat b)\right)d\nu\enspace.$$ Note that
$\al^+=\al^-$ everywhere except on the set $X$, so the integrands are only nonzero on $X$,
thus we can replace $\Re^n$ with $X$ in the integrals: $$\int_{\hat b\in
X}\left(g(\al^+(\hat b),\hat b)-g(\al^-(\hat b),\hat
b)\right)d\mu\TCSdd=\int_{\hat b\in X}\left(A_i^+(\hat b)-\al_i^-(\hat
b)\right)d\nu\enspace.$$

Now, note that on $X$, $\al^+=\al$ and $\al^-=0$. Thus, also using the fact
$g(\al^-(\hat b),\hat b)=0$, we have $$\int_{\hat b\in X}g(\al(\hat b),\hat
b)d\mu\TCSdd=\int_{\hat b\in X}A_i(\hat b)d\nu\enspace,$$ as desired.

To show that the lemma holds for all Borel measurable sets $X$, we observe that it holds for
all rectangular parallelepipeds (a product of open,
closed, or half-open intervals) by the above argument. Since the set of rectangular parallelepipeds is closed
under finite intersections, the lemma applies to all finite intersections of
rectangular parallelepipeds, which is the $\pi$-system that generates the Borel $\sigma$-algebra
of $\Re^n$.

Additionally, if the lemma holds for a countable sequence of disjoint sets $X_k$, then it
clearly holds for their union as well, implying that the sets for which the lemma is true
must be a $\lambda$-system.

Therefore, by Dynkin's $\pi$-$\lambda$ theorem, the $\lambda$-system (the sets satisfying
the lemma) must contain all sets in the $\sigma$-algebra generated by the $\pi$-system (the set
of rectangular parallelepipeds) --- namely, it must contain all sets in the Borel $\sigma$-algebra
of $\Re^n$. Thus, the lemma must hold for all Borel measurable sets $X$.
\end{proof}

\subsection{Monotonicity and $\mub$}
Theorem~\ref{thm:a-sp-char} requires $\mub$ to be such that $\asci(b)$ is monotone in expectation. The following lemma gives a necessary condition:
\begin{lemma}
\label{lem:a-sp-char-mono}
Let $B$ be a set of bids that is {\em leftward closed} with respect to $b_i$, i.e. if $\hat b\in B$, then $(u,\hat b_{-i})\in B$ for all $u\in (-\infty,b_i]\cap\Tyi$. If $\mub(B)$ satisfies the monotonicity condition
$$\macroCPr{\hat b\in B}{b}=\mub(B)$$
is weakly decreasing in $b_i$. Similarly, if $B$ is {\em rightward closed} with respect to $b_i$ (i.e. $\hat b\in B$ implies $\hat b_{-i}u\in B$ for $u\in[b_i,\infty)$), then $\pr({\hat b\in B}|{b})$ is weakly increasing in $b_i$, and if $B$ is both rightward and leftward closed with respect to $b_i$ then $\pr({\hat b\in B}|{b})$ is constant in $b_i$.
\end{lemma}

\begin{proof} First, we prove the case where $B$ rightward closed. For contradiction, let $B$ be a rightward closed set on which $f$ violates the statement of the lemma for some $b$ and $b_i'>b_i$, i.e.
$$\macroCPr{\hat b\in B}{b}=\mub(B)>\mu_{b_i',b_{-i}}(B)=\macroCPr{\hat b\in B}{b_i',b_{-i}}\enspace.$$
Consider the monotone allocation function
\begin{equation*}
A_i(b)=\begin{cases}
1,&b\in B\\
0,&otherwise.
\end{cases}
\end{equation*}
Noting that the $\ex[\asc(b)]=\mub(B)$, we have
$$E[\asci(b_i',b_{-i})]=\mub[b_i',b_{-i}](B)<\mub(B)=E[\asci(b)]\enspace.$$
Thus, under this allocation function, bidder $i$ lowers her expected utility by raising her bid to $b_i'$, contradicting the monotonicity condition.

Finally, any leftward closed set $B$ is the complement (probabilistically) of a rightward closed set, therefore $\pr({\hat b\in B}|{b})$ must be weakly decreasing. For a set $B$ that is both leftward and rightward closed, the theorem follows because $\pr({\hat b\in B}|{b})$ must be both weakly increasing and weakly decreasing.
\end{proof}

\section{Optimality proofs for generalized BKS}

\label{sec:a-babaioff}

In this section, we generalize our optimality result of Section~\ref{sec:babaioff} to arbitrary probability measures and give a complete proof. Theorem~\ref{thm:a-sp-char} shows that truthful payments take the form
\[\lambda_i(\al(\hat b),\hat b,b)=\rhob(\hat b)\al_i(\hat b)+\lambda_i^0(\hat b,b)\enspace a.s.\]
and thus optimizing the bid-normalized payments means optimizing the following quantity:
\[\sum_j \frac{\lambda_{ij}(b_j(\al(\hat
b)),\hat b,b)}{b_j(\al(\hat b))}=\frac{\rhob(\hat b)\al_i(\hat b)}{b_i\al_i(\hat b)}=\frac{\rhob(\hat b)}{b_i}\enspace.\]
This means that for worst-case payments we will optimize $\sup_{i,\hat b}\left|\frac{\rhob(\hat b)}{b_i}\right|$, and for payment variance we will optimize $\max_i\var_{\hat b\sim\mub}\left(\frac{\rhob(\hat b)}{b_i}\right)$.
We show that the BKS transformation is optimal for both, subject to an almost everywhere caveat:
\begin{theorem}[Optimality of the BKS Transformation]
\label{thm:a-babaioff-opt}
(Generalization of Theorem~\ref{thm:babaioff-opt})
The BKS reduction $\BKSRedn(\al,\gamma)$ optimizes the payment variance and worst-case normalized payment subject to a lower bound of $\alpha=(1-\gamma)^n\in(\frac1e,1)$ on the precision, the welfare approximation ($n\geq2$), or the revenue approximation ($n\geq2$). That is, for any other truthful reduction $(\mu,\{\lambda_i\})$ that achieves a precision, welfare approximation, or revenue approximation of $\alpha$, the worst-case normalized payments are at least as large almost everywhere over $b$:
\[\sup_{\al,i}\var_{\hat b\sim\mub}\left(\sum_j \frac{\lambda_{ij}(b_j(\al(\hat
b)),\hat b,b)}{b_j(\al(\hat b))}\right)=\max_{i}\var_{\hat b\sim\mub}\left(\frac{\rhob(\hat b)}{b_i}\right)\geq\max_{i}\var_{\hat b\sim\mub}\left(\frac{\rhob[\fbks](\hat b)}{b_i}\right)\enspace a.e.\]
and
\[\sup_{\al,i,\hat b}\left|\sum_j \frac{\lambda_{ij}(b_j(\al(\hat
b)),\hat b,b)}{b_j(\al(\hat b))}\right|=\sup_{i,\hat b}\left|\frac{\rhob(\hat b)}{b_i}\right|\geq\sup_{i,\hat b}\left|\frac{\rhob[\fbks](\hat b)}{b_i}\right|\enspace a.e.\]
Under the nice distribution assumption, this holds for every $b$.
\end{theorem}

The theorem is proven in two steps. First, we prove in Theorem~\ref{thm:a-babaioff-pre} that the BKS transform optimizes precision. Second, we show in Lemma~\ref{lem:a-sp-swrapx} that a distribution which optimizes precision also optimizes the welfare and revenue approximations.

\begin{theorem}[Precision Optimality of the BKS Transformation]
\label{thm:a-babaioff-pre}
(Generalization of Theorem~\ref{thm:babaioff-opt})
The BKS reduction $\BKSRedn(\al,\gamma)$ optimizes the variance of normalized payments and the worst-case normalized payment subject to a lower bound of $\alphap=(1-\gamma)^n\in(\frac1e,1)$ on the precision almost everywhere over $b$. Under the nice distribution assumption, it is optimal for every $b$.
\end{theorem}
Theorem~\ref{thm:a-babaioff-pre} is given in Sections~\ref{sec:a-babaioff-def}-\ref{sec:a-babaioff-tech}. Section~\ref{sec:a-babaioff-def} defines probabilities that are used in the proof. Section~\ref{sec:a-babaioff-atb} proves Theorem~\ref{thm:a-babaioff-pre} with forward references to two important technical lemmas given in Section~\ref{sec:a-babaioff-tech}.

\begin{lemma}\label{lem:a-sp-swrapx} (Generalization of Lemma~\ref{lem:sp-swrapx}) For $\alpha>\frac1{e}$ and $n\geq2$, a probability measure that optimizes the variance of normalized payments or the maximum normalized payment subject to a precision constraint of $\Pr(\hat b= b|b)\geq\alpha$ also optimizes the maximum normalized payment almost everywhere subject to a welfare or revenue approximation of $\alpha$.
\end{lemma}

Lemma~\ref{lem:a-sp-swrapx} is proven in Section~\ref{sec:a-sp-swrapx}, building on technical lemmas form Section~\ref{sec:a-babaioff-tech}.

\subsection{Definitions}
\label{sec:a-babaioff-def}

To prove Theorem~\ref{thm:a-babaioff-pre}, we give names to certain probabilities. As in the MIDR setting, we use a set $M\subseteq[n]$ to denote the set of bidders with $\hat b_i=b_i$. Bidders $i\not\in M$ have their bids lowered, that is $\hat b_i<b_i$. We define the probability $\pi^\mu(M,b)$ to be the probability that such an event occurs, that is, $\pi^\mu(M,b)$ is the probability when $b$ is bid that $\hat b_i=b_i$ if $i\in M$, and $\hat b_i<b_i$ if $i\not\in M$:
\[\pi^\mu(M,b)\equiv\macroCPr{(\hat b_i=b_i\mbox{ for }i\in M)\mbox{ and }(\hat b_i<b_i\mbox{ for }i\not\in M)}{b}\enspace.\]
Note that for the BKS transformation, $\pi^\mu(M,b)=(1-\gamma)^{|M|}\gamma^{n-|M|}$ so $\frac{\pi^\mu(M\cup\{i\},b)}{\pi^\mu(M,b)}=\frac{1-\gamma}{\gamma}$.

The second probability quantifies the behavior of $\mub$ near $b$ as follows. Fix a bid $b$ and assume player $i$ actually bids $b_i-\delta$. Does the distribution $\mub[b_i-\delta,b_{-i}]$ cause the reduction to select $\hat b=b$ with positive probability in spite of the fact that $i$ said $b_i-\delta$? In particular, we care about the average behavior for $\delta\in[0,b_i]$, which we represent by $z^\mu(M,i,\bar b)$. Formally, we define
\[\zeta^\mu(M,i,b,z)\equiv\macroCPr{\hat b_i=z\mbox{ and }(\hat b_j=b_j\mbox{ for }j\in M\setminus\{i\})\mbox{ and }(\hat b_j<b_j\mbox{ for }j\not\in M\cup\{i\})}{b}\]
and
\[z^\mu(M,i,b)\equiv\frac1{b_i}\int_0^{b_i}\zeta^\mu(M,i,(u,b_{-i}),b_i)\enspace.\]
Of particular importance, we will show $z^\mu(M,i,b)=0$ almost everywhere in general and everywhere under the nice distribution assumption.

\subsection{Precision Optimality of the BKS Transformation}
\label{sec:a-babaioff-atb}

The optimality proof for the BKS transformation 

The first result follows as a corollary of Lemma~\ref{lem:babaioff-monoratio}:
\begin{corollary}
\label{cor:babaioff-piratio} (of Lemma~\ref{lem:babaioff-monoratio})
If a resampling distribution $\mu$ satisfies the monotonicity condition, then for all $M$, $i\not\in M$:
\[\sup_{\hat b} \left|\frac{\rhob(\hat b)}{b_i}\right|\geq\frac{\pi^\mu(M\cup\{i\},b)-z^\mu(M,i,b)}{\pi^\mu(M,b)}\]
and
\begin{align*}
&\int_{\hat b_i\leq b_i\wedge(j\in M\Rightarrow\hat b_j=b_j)\wedge(j\not\in M\cup\{i\}\Rightarrow\hat b_j<b_j)}\left(\frac{\rhob(\hat b)}{b_i}\right)^2d\mub\\
&\quad\geq\left(\pi^\mu(M,b)+\pi^\mu(M\cup\{i\},b)\right)\frac{\pi^\mu(M\cup\{i\},b)}{\pi^\mu(M,b)}\left(1-\frac{z^\mu(M,i,b)}{\pi^\mu(M\cup\{i\},b)}\right)^2\enspace.
\end{align*}
\end{corollary}
\begin{proof} Apply Lemma~\ref{lem:babaioff-monoratio} where $B_{-i}$ is the set of $\hat b_{-i}$ where $\hat b_j=b_j$ if $j\in M$ and $\hat b_j<b_j$ for $j\not\in M$.
\end{proof}

If we ignore the $z^\mu(M,i,b)$ terms, this looks precisely like the normalized payments from the MIDR setting. Fortunately, $z^\mu(M,i,b)$ is almost always zero:
\begin{corollary}
\label{cor:babaioff-zae} (of Lemma~\ref{lem:babaioff-ae})
For any resampling distribution $\mu$ and a fixed $M$ and $i$,
\[z^\mu(M,i,b)=0\enspace a.e.\]
(i.e. for all but a set of $b$ with zero measure).

Under the nice distribution assumption, $z^\mu(M,i,b)=0$ for all $b$.
\end{corollary}
\begin{proof}
Note that $\zeta^\mu(M,i,(u,b_{-i}),b_i)\leq\Pr_\mu(\hat b_i=b_i|u,b_{-i})$, so by Lemma~\ref{lem:babaioff-ae}
\[z^\mu(M,i,b)=\frac1{b_i}\int_0^{b_i}\zeta^\mu(M,i,(u,b_{-i}),b_i)\leq\int_0^{b_i}\Pr_\mu(\hat b_i=b_i|u,b_{-i})=0\enspace a.e.\]
\end{proof}

Thus, Corollaries~\ref{cor:babaioff-piratio} and~\ref{cor:babaioff-zae} together imply the following bound:
\begin{lemma}
\label{lem:babaioff-relgamma}
If a resampling disrtibution $\mu$ with precision $\alpha\geq(1-\gamma)^n$ satisfies the monotonicity condition, then
\[\sup_{i,\hat b} \left|\frac{\rhob(\hat b)}{b_i}\right|\geq\frac{1-\gamma}{\gamma}\enspace a.e.\]
that is, for all $b$ but a set with measure zero. This holds everywhere if $z^\mu(M,i,b)=0$ everywhere.
\end{lemma}
\begin{proof}
We first prove the bound on the worst-case normalized payment. By assumption on the precision of $\mu$, we have $\pi^\mu([n],b)\geq(1-\gamma)^n$ for some $\gamma$ and all $b$. By Corollary~\ref{cor:babaioff-piratio}, we know that
\[\sup_{\hat b} \left|\frac{\rhob(\hat b)}{b_i}\right|\geq\frac{\pi^\mu(M\cup\{i\},b)-z^\mu(M,i,b)}{\pi^\mu(M,b)}\enspace.\]
Applying Lemma~\ref{lem:babaioff-hfunc} with $\eta(S)=\pi^\mu(S,b)$, $\alpha=(1-\gamma)^n$, and $\beta=1$ we get that
\[\max_{M,i\not\in M}\frac{\pi^\mu(M\cup\{i\},b)}{\pi^\mu(M,b)}\geq\frac{1-\phi}{\phi}\]
where
\[\phi=1-\left(\frac{(1-\gamma)^n}{1}\right)^{\frac1n}=\gamma\enspace.\]
Thus,
\[\max_{M,i\not\in M}\frac{\pi^\mu(M\cup\{i\},b)}{\pi^\mu(M,b)}\geq\frac{1-\gamma}{\gamma}\enspace.\]

Aggregating Corollary~\ref{cor:babaioff-piratio} over all $M$ and $i\not\in M$, we have

\begin{align*}
\sup_{i,\hat b} \left|\frac{\rhob(\hat b)}{b_i}\right|&\geq\max_{M,i\not\in M}\frac{\pi^\mu(M\cup\{i\},b)}{\pi^\mu(M,b)}-\max_{M,i\not\in M}\frac{z^\mu(M,i,b)}{\pi^\mu(M,b)}\\
&\geq\frac{1-\gamma}{\gamma}-\max_{M,i\not\in M}\frac{z^\mu(M,i,b)}{\pi^\mu(M,b)}\enspace.
\end{align*}
If we assume $z^\mu(M,i,b)=0$ everywhere (e.g. by the nice distribution assumption), then we get
\[\sup_{i,\hat b} \left|\frac{\rhob(\hat b)}{b_i}\right|\geq\frac{1-\gamma}{\gamma}\enspace.\]
Otherwise, Corollary~\ref{cor:babaioff-zae} says that $z^\mu(M,i,b)=0$ almost everywhere, giving the more general bound.
\end{proof}

\begin{lemma}
\label{lem:babaioff-vargamma}
If a resampling distribution $\mu$ with precision $\alpha\geq(1-\gamma)^n\geq\frac1e$ satisfies the monotonicity condition, then
\[\max_i\var_{\hat b}\left(\frac{\rhob(\hat b)}{b_i}\right)\geq\frac{1-\gamma}{\gamma} \enspace a.e.\]
that is, for all $b$ but a set with measure zero. This holds everywhere if $z^\mu(M,i,b)=0$ everywhere.
\end{lemma}
\begin{proof} The proof for variance is similar to Lemma~\ref{lem:babaioff-relgamma}, but we apply Lemma~\ref{lem:babaioff-hfunc2} instead of Lemma~\ref{lem:babaioff-hfunc}. First, note that since $\mub$ is a probability measure, $\mub(\Re^n)=1$ and thus
\[\int_{\hat b\in\Re^n}\frac{\rhob(\hat b)}{b_i}d\mub=\frac1{b_i}\mubpsci(\Re^n)=\mub(\Re^n)-\frac1{b_i}\int_0^{b_i}\mub[u,b_{-i}](\Re^n)du=1-\frac1{b_i}\int_0^{b_i}1du=0\enspace.\]
We begin with the variance for player $i$, applying Corollaries~\ref{cor:babaioff-piratio} and~\ref{cor:babaioff-zae}:
\begin{align*}
\var_{\hat b}\left(\frac{\rhob(\hat b)}{b_i}\right)&=\int_{\hat b\in\Re^n}\left(\frac{\rhob(\hat b)}{b_i}\right)^2d\mub-\left(\int_{\hat b\in\Re^n}\frac{\rhob(\hat b)}{b_i}d\mub\right)^2\\
&=\int_{\hat b\in\Re^n}\left(\frac{\rhob(\hat b)}{b_i}\right)^2d\mub\\
&\geq\sum_{M|i\not\in M}\int_{\hat b_i\leq b_i\wedge(j\in M\Rightarrow\hat b_j=b_j)\wedge(j\not\in M\cup\{i\}\Rightarrow\hat b_j<b_j)}\left(\frac{\rhob(\hat b)}{b_i}\right)^2d\mub\\
&\geq\sum_{M|i\not\in M}\left(\pi^\mu(M,b)+\pi^\mu(M\cup\{i\},b)\right)\frac{\pi^\mu(M\cup\{i\},b)}{\pi^\mu(M,b)}\enspace a.e.
\end{align*}
Applying Lemma~\ref{lem:babaioff-hfunc2} with $\eta(S)=\pi^\mu(S,b)$, $\alpha=(1-\gamma)^n$ and $\beta=\macroCPr{\hat b\leq b}{b}$ immediately implies
\[\max_i\var_{\hat b}\left(\frac{\rhob(\hat b)}{b_i}\right)\geq\macroCPr{\hat b\leq b}{b}\frac{1-\phi}{\phi} \enspace a.e.\]
where
\[\phi=1-\left(\frac{(1-\gamma)^n}{\macroCPr{\hat b\leq b}{b}}\right)^{\frac1n}\enspace.\]
One can check that when $\frac{(1-\gamma)^n}{\macroCPr{\hat b\leq b}{b}}\geq\frac1e$, the quantity $\macroCPr{\hat b\leq b}{b}\frac{1-\phi}{\phi} $ is decreasing in $\macroCPr{\hat b\leq b}{b}$. Taking the worst case $\macroCPr{\hat b\leq b}{b}=1$ implies the desired result:
\[\max_i\var_{\hat b}\left(\frac{\rhob(\hat b)}{b_i}\right)\geq\max_i\geq\frac{1-\gamma}{\gamma} \enspace a.e.\]
\end{proof}

Theorem \ref{thm:a-babaioff-pre} -- optimality of the BKS transformation with respect to a precision bound -- follows from the two previous lemmas:

\begin{proof}[ of Theorem \ref{thm:a-babaioff-pre}] For worst-case payments, we show that for any measure $\mu$, with precision at least $2^{-n}$,
\[\sup_{i,\hat b} \left|\frac{\rhob[\fbks](\hat b)}{b_i}\right|\leq\sup_{i,\hat b} \left|\frac{\rhob(\hat b)}{b_i}\right|\enspace a.e.\]
For $\Pr(\hat b=b|b)=(1-\gamma)^n$, the BKS transform achieves $\sup_{i,\hat b} \left|\frac{\rhob[\fbks](\hat b)}{b_i}\right|=\max\left(1,\frac{1-\gamma}{\gamma}\right)$ for all $b$. Provided $\gamma>\frac12$, the dominant term is $\frac{1-\gamma}{\gamma}$ and Lemma~\ref{lem:babaioff-relgamma} shows that this is a lower bound for any such $\mu$ almost everywhere. When $\alpha>2^{-n}$ we get $\gamma>\frac12$, and thus $\fbks$ is optimal.

Moreover, under the nice distribution assumption (implying $z^\mu(M,i,b)=0$), Lemma~\ref{lem:babaioff-relgamma} says that this holds everywhere.

For the variance of normalized payments, we need to show that for any measure $\mu$ with precision at least $\frac1e$:
\[\var_{\hat b\sim\mub} \left(\frac{\rhob[\fbks](\hat b)}{b_i}\right)\leq\var_{\hat b\sim\mub} \left(\frac{\rhob(\hat b)}{b_i}\right)\enspace a.e.\]
Again, for $\Pr(\hat b=b|b)=(1-\gamma)^n$, the BKS transform achieves $\var_{\hat b\sim\mub} \left|\frac{\rhob[\fbks](\hat b)}{b_i}\right|=\frac{1-\gamma}{\gamma}$ for all $b$. Lemma~\ref{lem:babaioff-vargamma} shows that this is a lower bound for any such $\mu$ almost everywhere.
\end{proof}

\subsection{Technical Lemmas}
\label{sec:a-babaioff-tech}

The next lemma gives our main lower bound on the worst coefficient:

\begin{lemma}
\label{lem:babaioff-monoratio}
If a measure $\mu$ satisfies the monotonicity condition, then for any player $i$, bid $b$, and set of bids $B_{-i}\subseteq\Re^{n-1}_+$:
\[\sup_{\hat b} \left|\frac{\rhob(\hat b)}{b_i}\right|\geq\frac{\macroCPr{\hat b_i=b_i\wedge\hat b_{-i}\in B_{-i}}{b}-\frac1{b_i}\int_0^{b_i}\macroCPr{\hat b_i=b_i\wedge\hat b_{-i}\in B_{-i}}{u,b_{-i}}}{\macroCPr{\hat b_i< b_i\wedge\hat b_{-i}\in B_{-i}}{b}}\enspace,\]
and
\begin{align*}
\int_{\hat b_i\leq b_i\wedge\hat b_{-i}\in B_{-i}}\left(\frac{\rhob(\hat b)}{b_i}\right)^2d\mub\geq&{}\macroCPr{\hat b_i\leq b_i\wedge\hat b_{-i}\in B_{-i}}{b}\frac{\macroCPr{\hat b_i=b_i\wedge\hat b_{-i}\in B_{-i}}{b}}{\macroCPr{\hat b_i<b_i\wedge\hat b_{-i}\in B_{-i}}{b}}\\
&{}\times\left(1-\frac{\frac1{b_i}\int_0^{b_i}\macroCPr{\hat b_i=b_i\wedge\hat b_{-i}\in B_{-i}}{u,b_{-i}}}{\macroCPr{\hat b_i=b_i\wedge\hat b_{-i}\in B_{-i}}{b}}\right)^2\enspace,
\end{align*}
where the integral terms are zero almost everywhere in $b$ by Lemma~\ref{lem:babaioff-ae}.
\end{lemma}

\begin{proof} Define the sets
\[B^{(=)}=\{b_i\}\times B_{-i}\quad\mbox{and}\quad B^{(<)}=[0,b_i)\times B_{-i} \enspace,\]
i.e. the set $B^{(=)}$ contains bids $\hat b$ where $\hat b_i=b_i$ and $\hat b_{-i}\in B_{-i}$, and the set $B^{(<)}$ contains bids $\hat b$ where $\hat b_i< b_i$ and $\hat b_{-i}\in B_{-i}$.
The main work of the lemma is to bound the following term:
\begin{align*}
\int_{\hat b\in B^{(<)}}\frac{\rhob(\hat b)}{b_i}d\mub&=\frac{\mubpsci(B^{(<)})}{b_i}\\
&=\frac{b_i\mub(B^{(<)})-\int_0^{b_i}\mub[u,b_{-i}](B^{(<)})du}{b_i}\\
&=\mub(B^{(<)})-\frac{1}{b_i}\int_0^{b_i}\mub[u,b_{-i}](B^{(<)})du\\
&=\macroCPr{\hat b\in B^{(<)}}{b}-\frac1{b_i}\int_{0}^{b_i}\macroCPr{\hat b\in B^{(<)}}{u,b_{-i}}du\enspace.
\end{align*}
By monotonicity, $\macroCPr{\hat b\in B^{(<)}\cup B^{(=)}}{u,b_{-i}}$ is weakly decreasing in $u$ (Lemma~\ref{lem:a-sp-char-mono}). This implies
\[\macroCPr{\hat b\in B^{(=)}}{b}+\macroCPr{\hat b\in B^{(<)}}{b}\leq\frac1{b_i}\int_{0}^{b_i}\left(\macroCPr{\hat b\in B^{(=)}}{u,b_{-i}}+\macroCPr{\hat b\in B^{(<)}}{u,b_{-i}}\right)du\]
and thus
\begin{align*}
\int_{\hat b\in B^{(<)}}\frac{\rhob(\hat b)}{b_i}d\mub&=\macroCPr{\hat b\in B^{(<)}}{b}-\frac1{b_i}\int_{0}^{b_i}\macroCPr{\hat b\in B^{(<)}}{u,b_{-i}}du\\
&\leq-\left(\macroCPr{\hat b\in B^{(=)}}{b}-\frac1{b_i}\int_{0}^{b_i}\macroCPr{\hat b\in B^{(=)}}{u,b_{-i}}du\right)\enspace.
\end{align*}

To bound $\sup_{\hat b\in B^{(<)}}\left|\frac{\rhob(\hat b)}{b_i}\right|$, we have
\[\sup_{\hat b\in B^{(<)}}\left|\frac{\rhob(\hat b)}{b_i}\right|\geq\left|\frac{\int_{\hat b\in B^{(<)}}\frac{\rhob(\hat b)}{b_i}d\mub}{\mub(B^{(<)})}\right|\geq\frac{\macroCPr{\hat b\in B^{(=)}}{b}-\frac1{b_i}\int_0^{b_i}\macroCPr{\hat b\in B^{(=)}}{u,b_{-i}}}{\macroCPr{\hat b\in B^{(<)}}{b}}\]
Lemma \ref{lem:babaioff-ae} implies that the limit term is zero almost everywhere in $b$.

For our partial bound on the second moment, we write
\begin{align*}
\int_{\hat b\in B^{(<)}\cup B^{(=)}}\left(\frac{\rhob(\hat b)}{b_i}\right)^2d\mub\geq&{}\int_{\hat b\in B^{(=)}}\left(\frac{\rhob(\hat b)}{b_i}\right)^2d\mub+\int_{\hat b\in B^{(<)}}\left(\frac{\rhob(\hat b)}{b_i}\right)^2d\mub\\
\geq&{}\mub(B^{(=)})\left(\frac{\int_{\hat b\in B^{(=)}}\frac{\rhob(\hat b)}{b_i}d\mub}{\mub(B^{(=)})}\right)^2+\mub(B^{(<)})\left(\frac{\int_{\hat b\in B^{(<)}}\frac{\rhob(\hat b)}{b_i}d\mub}{\mub(B^{(<)})}\right)^2\\
\geq&{}\mub(B^{(=)})\left(\frac{\macroCPr{\hat b\in B^{(=)}}{b}-\frac1{b_i}\int_{0}^{b_i}\macroCPr{\hat b\in B^{(=)}}{u,b_{-i}}du}{\mub(B^{(=)})}\right)^2\\
&{}+\mub(B^{(<)})\left(\frac{\macroCPr{\hat b\in B^{(=)}}{b}-\frac1{b_i}\int_{0}^{b_i}\macroCPr{\hat b\in B^{(=)}}{u,b_{-i}}du}{\mub(B^{(<)})}\right)^2\\
\geq&{}\left(\mub(B^{(<)})+\mub(B^{(=)})\right)\frac{\mub(B^{(=)})}{\mub(B^{(<)})}\\
&{}\times\left(1-\frac{\frac1{b_i}\int_{0}^{b_i}\macroCPr{\hat b\in B^{(=)}}{u,b_{-i}}du}{\macroCPr{\hat b\in B^{(=)}}{b}}\right)^2
\end{align*}
Which is the desired bound.
\end{proof}

\begin{lemma}
\label{lem:babaioff-hfunc}
Let $\eta:\{0,1\}^n$ be a function over subsets $S\subseteq[n]$ with $\eta([n])\geq\alpha\in[0,1]$ and $\sum_{S\subseteq[n]}\eta(S)\leq\beta\in[0,1]$. Then
\[\max_{S,i\in[n]\setminus S}\frac{\eta(S\cup\{i\})}{\eta(S)}\geq\frac{1-\phi}{\phi}\]
where $\phi=1-\left(\frac\alpha\beta\right)^{\frac1n}$.
\end{lemma}
\begin{proof} By contradiction. Assume that for every $S$ and $i\not\in S$,
\[\frac{\eta(S\cup\{i\})}{\eta(S)}<\frac{1-\phi}{\phi}\]
where $\phi=1-\left(\frac\alpha\beta\right)^{\frac1n}$.

Then by multiplying $\frac{\eta(S)}{\eta(S\cup\{i\})}$ terms together we get
\[\eta(S)\geq\eta([n])\left(\frac{\phi}{1-\phi}\right)^{n-|S|}\enspace.\]
Summing over all $S\subseteq[n]$, substituting for $\alpha$ and $\beta$, and algebra gives
\begin{align*}
\sum_{S\subseteq[n]}\eta(S)&>\eta([n])\sum_{S\subseteq[n]}\left(\frac{\phi}{1-\phi}\right)^{n-|S|}\\
\beta&>\alpha\sum_{S\subseteq[n]}\left(\frac{\phi}{1-\phi}\right)^{n-|S|}\\
\beta(1-\phi)^n&>\alpha\sum_{S\subseteq[n]}(1-\phi)^{|S|}\phi^{n-|S|}\\
\beta\left(1-\phi\right)^n&>\alpha\\
\beta\left(1-\left(1-\left(\frac\alpha\beta\right)^{\frac1n}\right)\right)^n&>\alpha\\
\alpha&>\alpha\enspace.
\end{align*}
Which is a contradiction.
\end{proof}

\begin{lemma}
\label{lem:babaioff-hfunc2}
Let $\eta:\{0,1\}^n$ be a function over subsets $S\subseteq[n]$ with $\eta([n])\geq\alpha\in[0,1]$ and $\sum_{S\subseteq[n]}\eta(S)=\beta\in[0,1]$. Then
\[\max_i\sum_{S|i\not\in S}\left(\eta(S)+\eta(S\cup\{i\})\right)\frac{\eta(S\cup\{i\})}{\eta(S)}\geq\beta\frac{1-\phi}{\phi}\]
where $\phi=1-\left(\frac\alpha\beta\right)^{\frac1n}$.
\end{lemma}
\begin{proof} We lower-bound the sum. Fix $i$ and differentiate the sum:
\[\frac{\partial}{\partial\eta(S)}\left(\sum_{T|i\not\in T}\left(\eta(T)+\eta(T\cup\{i\})\right)\frac{\eta(T\cup\{i\})}{\eta(T)}\right)=\begin{cases}
2\frac{\eta(S)}{\eta(S\setminus\{i\})}+1,&i\in S\\
-\left(\frac{\eta(S\cup\{i\})}{\eta(S)}\right)^2,&i\not\in S\enspace.
\end{cases}\]
The conditions of the lemma bound $\sum_S\eta(S)$ and $\eta([n])$, otherwise the values of $\eta$ are only constrained to be in $[0,1]$.  The derivative tells us that in an optimal assignment, for all sets $S$ that do not contain $i$, the ratio $\frac{\eta(S\cup\{i\})}{\eta(S)}$ is constant. Construct such an optimal assignment and define $\phi_i$ as satisfying
\[\frac{\eta(S\cup\{i\})}{\eta(S)}=\frac{1-\phi_i}{\phi_i}\]
for all $S$ that do not contain $i$. Note that this implies
\[\sum_{S|i\not\in S}\left(\eta(S)+\eta(S\cup\{i\})\right)\frac{\eta(S\cup\{i\})}{\eta(S)}\geq\beta\frac{1-\phi_i}{\phi_i}\]

For any set $S$ it follows that
\begin{align*}
\eta(S)&=\eta([n])\prod_{i\not\in S}\frac{\phi_i}{1-\phi_i}\\
\sum_{S\subseteq[n]}\eta(S)&=\eta([n])\sum_{S\subseteq[n]}\prod_{i\not\in S}\frac{\phi_i}{1-\phi_i}\\
\beta\prod_{i\in[n]}(1-\phi_i)&\geq\alpha\sum_{S\subseteq[n]}\prod_{i\in S}(1-\phi_i)\prod_{i\not\in S}\phi_i\\
\prod_{i\in[n]}(1-\phi_i)&\geq\frac\alpha\beta\enspace.
\end{align*}
This implies there is some $i$ such that $\phi_i\leq1-\left(\frac\alpha\beta\right)^{\frac1n}$, which implies the lemma.
\end{proof}

The next lemma is our main analysis lemma. We will ultimately use it to claim that our lower bound must hold almost everywhere for any $\mu$:

\begin{lemma}
\label{lem:babaioff-ae}
For any resampling distribution $\mu$ that satisfies the monotonicity condition, any bid $b$, and any bidder $i$,
\[\int_0^{b_i}\Pr_\mu(\hat b_i=b_i|u,b_{-i})=0\enspace a.e.\]
(i.e. for all but a set of $b$ with zero measure).
\end{lemma}
\begin{proof} Define the marginalized measure $\mub^i$ for a set of bids $B\subseteq\Re$ as
\[\mub^i(B)\equiv\mub(\{b\in\Re^n|b_i\in B\})\enspace.\]
Note that
\[\mub[u,b_{-i}]^i(\{b_i\})=\Pr_\mu(\hat b_i=b_i|u,b_{-i})\]
and therefore our task is to show that
\[\lim_{u\rightarrow{}^-b_i}\mub[u,b_{-i}]^i(\{b_i\})=0\enspace a.e.\]

Next we show that for any $b$ we can prove the desired limit is zero by proving that a related integral is zero. Assume that for some $b$ we have
\[\lim_{u\rightarrow{}^-b_i}\mub[u,b_{-i}]^i(\{b_i\})>0\enspace.\]
Then there exists a $\delta_b$ such that
\[\forall u\in(b_i-\delta_b,b_i):\quad \mub[u,b_{-i}]^i(\{b_i\})>0\enspace.\]
Since $\mub[u,b_{-i}]^i(\{b_i\})$ is nonnegative, this implies
\[\int_{u\in\Re}\mub[u,b_{-i}]^i(\{b_i\})du\geq\int_{u\in(b_i-\delta_b,b_i)}\mub[u,b_{-i}]^i(\{b_i\})du>0\enspace.\]
Taking the contrapositive, it follows that if the integral is zero at a bid $b$ then the limit is also zero:
\begin{equation}
\label{eqn:babaioff-ae-imp}
\int_{u\in\Re}\mub[u,b_{-i}]^i(\{b_i\})du=0\Rightarrow\lim_{u\rightarrow{}^-b_i}\mub[u,b_{-i}]^i(\{b_i\})=0\enspace.
\end{equation}
Henceforth, we will prove that $\int_{u\in\Re}\mub[u,b_{-i}]^i(\{b_i\})du=0$ almost everywhere.

We start with the integral
\begin{equation*}
\label{eqn:babaioff-ae-int}
\int_{b\in\Re}\int_{u\in\Re}\mub[u,b_{-i}]^i(\{b_i\})du db\enspace.
\end{equation*}
Manipulating the integral and noting that $\int_{u\in\Re} 1_{\{u\}}(\hat b_i)du=0$, we get
\begin{align*}
\int_{b\in\Re}\int_{u\in\Re}\mub[u,b_{-i}]^i(\{b_i\})du db&=\int_{b\in\Re}\int_{u\in\Re}\mub[b_i,b_{-i}]^i(\{u\})du db\\
&=\int_{b\in\Re^n}\int_{u\in\Re}\int_{\hat b_i\in \Re} 1_{\{u\}}(\hat b_i)d\mub^i du db\\
&=\int_{b\in\Re^n}\int_{\hat b_i\in \Re}\int_{u\in\Re} 1_{\{u\}}(\hat b_i)du d\mub^i db\\
&=\int_{b\in\Re^n}\int_{\hat b_i\in \Re} 0 d\mub^i db\\
&=0
\end{align*}
(where integral rearrangements may be justified by Tonelli's Theorem). By Fact~\ref{fact:ana-nnpfzeroae}, this implies
\[\int_{u\in\Re}\mub[u,b_{-i}]^i(\{b_i\})du=0\enspace\mbox{almost everywhere over $b$,}\]
which implies the desired result.
\end{proof}

\subsection{Welfare and Revenue Optimality}
\label{sec:a-sp-swrapx}

Under mild assumptions, one can show that optimizing precision is equivalent to optimizing the social welfare approximation or the revenue approximation. We include only the worst-case optimality proofs; the variance proof is similar, applying ideas from Lemma~\ref{lem:babaioff-vargamma}.

The optimality proof is divided into two steps:
\begin{enumerate}
\item {\em Lemmas~\ref{lem:babaioff-wapx} and~\ref{lem:babaioff-rapx}:} Show that the welfare/revenue approximation of a resampling distribution $\mu$ is essentially
\[\inf_b\min_{i\in[n]}\macroCPr{\hat b_i\geq b_i\mbox{ and }\hat b_{-i}=b_{-i}}{b}\enspace.\]
The welfare and revenue lemmas use different techniques to give a lower bound on the approximation; however, they use the same ``bad'' allocation function.
\item {\em Lemma~\ref{lem:babaioff-relwrgamma} and finally Lemma~\ref{lem:a-sp-swrapx}:} Show that a distribution that optimizes the worst-case normalized payment with respect to
\[\min_{i\in[n]}\macroCPr{\hat b_i\geq b_i\mbox{ and }\hat b_{-i}=b_{-i}}{b}\geq\alpha\]
must take $\Pr(\hat b\not\leq b|b)=0$ and, therefore
\[\min_{i\in[n]}\macroCPr{\hat b_i\geq b_i\mbox{ and }\hat b_{-i}=b_{-i}}{b}=\macroCPr{\hat b=b}{b}\]
implying that it is sufficient to optimize with respect to $\Pr(\hat b=b|b)\geq(1-\gamma)^n=\alpha$.
\end{enumerate}

The following lemmas characterize the welfare and revenue approximations of the reduction generated by a resampling distribution $\mu$:

\begin{lemma}\label{lem:babaioff-wapx}
The welfare approximation of a resampling distribution $\mu$ for a bid $b$ is
\[\alpha=\min_{i\in[n]}\macroCPr{\hat b_i\geq b_i\mbox{ and }\hat b_{-i}=b_{-i}}{b}\enspace.\]
\end{lemma}
\begin{proof} For a bid $b$, define the set $B^i\subset\Re^n_+$ as
$$B^i=\{\hat b|\hat b_i\geq b_i\mbox{ and }\hat b_{-i}=b_{-i}\}\enspace.$$
Monotonicity of $\al$ requires that for all $u\geq b_i$,
$$\al_i(u,b_{-i})\geq\al_i(b)\enspace.$$
Thus, the allocation received by player $i$ under $\asc$ is at least
$$\macroCPr{\hat b_i\geq b_i\mbox{ and }\hat b_{-i}=b_{-i}}{b}\al_i(b)\TCSdd=\macroCPr{\hat b\in B^i}{b}\al_i(b)$$
and thus the social welfare is at least
\begin{eqnarray*}
\sum_{i\in[n]}b_i\asci(b)&\geq&\sum_{i\in[n]}b_i\macroCPr{\hat b\in B^i}{b}\asci(b)\\
&\geq&\min_{i\in[n]}\left(\macroCPr{\hat b\in B^i}{b}\right)\sum_{i\in[n]}b_i\asci(b)\enspace.
\end{eqnarray*}
This lower bound is tightin the following allocation rule
\[\al_i(\hat b)=\begin{cases}
1& i=j\mbox{ and }\hat b\in B^{i}\\
0& otherwise
\end{cases}\]
when $j=\argmin_{i\in[n]}b_i\Pr(\hat b\in B^i|b)$.
\end{proof}

\begin{lemma}\label{lem:babaioff-rapx} The revenue approximation $\alphar$ of a reduction given by a resampling distribution $\mu$ is bounded from below by the precision
\[\alphap=\inf_b\macroCPr{\hat b= b}{b}\leq\alphar\]
and above by
\[\alphar\leq\inf_{b}\min_{i\in[n]}\macroCPr{\hat b_i\geq b_i\wedge\hat b_{-i}=b_{-i}}{b}\enspace.\]
\end{lemma}
\begin{proof} To see that the precision $\alphap=\inf_b\macroCPr{\hat b= b}{b}$ is a lower bound on the revenue approximation, consider decomposing the mechanism produced by the reduction as follows: with probability $\alphap$, the mechanism uses the original allocation function, and with probability $1-\alphap$ it chooses an allocation function $\al^{rs}$ that resamples bids more frequently. Since prices are linear, the final expected price will be the weighted sum of the truthful prices for $\al$ and the truthful prices for $\al^{rs}$.

For positive types, revenue from both $\al$ and $\al^{rs}$ will be nonnegative, and the revenue of the resulting mechanism will be the weighted sum of the revenues from $\al$ and $\al^{rs}$. Thus, since $\al$ is chosen with probability $\alphap$, the revenue of their combination will be at least $\alphap$ times the revenue from $\al$.

Next we use the allocation function from Lemma~\ref{lem:babaioff-wapx} to give an upper bound. For clarity, we assume that the infimum in the bound of $\alpha$ is attained by some $b$. (The proof when the infimum is not attained is messier but fundamentally the same.) Let $b$ be a bid such that
$$\min_{i\in[n]}\macroCPr{\hat b_i\geq b_i\wedge\hat b_{-i}=b_{-i}}{b}=\alpha\enspace.$$
 Again, let $B^i\subset\Re^n_+$ be the set
$$B^i=\{\hat b|\hat b_i\geq b_i\mbox{ and }\hat b_{-i}=b_{-i}\}\enspace,$$
and consider following allocation function, where $j=\argmin_{i\in[n]} b_i\Pr(\hat b\in B^i|b)$:
$$\al_i(\hat b)=\begin{cases}
1& i=j\mbox{ and }\hat b\in B^{i}\\
0& otherwise.
\end{cases}$$
When this allocation function is implemented directly with the Archer-Tardos pricing rule, the revenue when bidders say $b$ will be
$$\sum_{i\in[n]}b_iA_i(b)-\int_{-\infty}^{b_i}A_i(u,b_{-i})du=b_j\enspace.$$
Now, for any single call reduction, the expected revenue will be
\begin{eqnarray*}
\sum_{i\in[n]}b_i\ex[A_i^{sc}(b)]-\int_{-\infty}^{b_i}\ex[A_i^{sc}(u,b_{-i})]du\TCSeas
&\leq&b_j\ex[A_j^{sc}(b)]\\
&=&b_j\macroCPr{\hat b\in B^j}{b}\enspace.
\end{eqnarray*}
Thus, the revenue approximation when players bid $b$ is at most $\Pr(\hat b\in B^{j}| b)$.
\end{proof}

\begin{lemma} \label{lem:babaioff-relwrgamma} The worst-case bid-normalized payment for a resampling distribution $\mu$ is at least
$$\sup_{\hat b} \left|\frac{\rhob(\hat b)}{b_i}\right|\geq\max\left(\frac{1-\gamma^{(=)}}{\gamma^{(=)}},\frac{1-\gamma^{(>)}}{\gamma^{(>)}}\right)\enspace a.e.$$
where
$$\gamma^{(=)}=1-\left(\frac{\Pr(\hat b=b|b)}{\Pr(\hat b\leq b|b)}\right)^{\frac1n}$$
and
$$\gamma^{(>)}=1-\left(\frac{\min_{i\in[n]}\Pr(\hat b_i> b_i\wedge\hat b_{-i}=b_{-i}|b)}{\frac1nPr(\hat b\not\leq b|b)}\right)^{\frac1{n-1}}\enspace.$$
The bound holds everywhere under the nice distribution assumption.
\end{lemma}

\begin{proof} For the sake of clarity, we assume the nice distribution assumption. The general case follows naturally by carrying extra terms through the analysis.

Corollary~\ref{cor:babaioff-piratio} says that for any $M\subset[n]$ and $i\not\in M$,
\[\sup_{\hat b} \left|\frac{\rhob(\hat b)}{b_i}\right|\geq\frac{\pi^\mu(M\cup\{i\},b)}{\pi^\mu(M,b)}\enspace.\]
Since $\sum_{M\subseteq[n]}\pi(M,\bar b)=\macroCPr{\hat b\leq\bar b}{\bar b}$, applying Lemma~\ref{lem:babaioff-hfunc} with $\eta(S)=\pi^\mu(M,b)$ implies that
\[\max_{M\subseteq[n]}\frac{\pi^\mu(M\cup\{i\},b)}{\pi^\mu(M,b)}\geq\frac{1-\gamma^{(=)}}{\gamma^{(=)}}\]
where $\gamma^{(=)}$ is
\[\gamma^{(=)}=1-\left(\frac{\Pr(\hat b=b|b)}{\Pr(\hat b\leq\bar b|\bar b)}\right)^{\frac1n}\enspace.\]
Thus, 
\[\sup_{\hat b} \left|\frac{\rhob(\hat b)}{b_i}\right|\geq\frac{1-\gamma^{(=)}}{\gamma^{(=)}}\enspace.\]

Next, define $\nu^\mu(M,j,b)$ as the probability that $\hat b_j>b_j$ while bids $i\neq j$ obey $M$ (that is, $\hat b_i=b_i$ for $i\in M$ and $\hat b_i<b_i$ if $i\not\in M$). Lemma~\ref{lem:babaioff-monoratio} implies that for all $j$, $M\subseteq[n]\setminus\{j\}$ and $i\not\in M\cup\{j\}$,
\[\sup_{\hat b} \left|\frac{\rhob(\hat b)}{b_i}\right|\geq\frac{\nu^\mu(M\cup\{i\},j,b)}{\nu^\mu(M,j,b)}\enspace a.e.\]
For any particular $j$, applying Lemma~\ref{lem:babaioff-hfunc} with $\eta(S)=\nu^\mu(S,j,b)$ as above implies that
\[\max_{M\subset[n]\setminus\{j\}}\frac{\nu^\mu(M\cup\{i\},j,b)}{\nu^\mu(M,j,b)}\geq\frac{1-\gamma^{(j)}}{\gamma^{(j)}}\]
where $\gamma^{(j)}$ is
\[\gamma^{(j)}=1-\left(\frac{\Pr(\hat b_j> b_j\wedge\hat b_{-j}=b_{-j}|b)}{\Pr(\hat b_j> b_j\wedge\hat b_{-j}\leq b_{-j}|b)}\right)^{\frac1{n-1}}\enspace.\]

Since the probabilities $\Pr(\hat b_j> b_j\wedge\hat b_{-j}\leq b_{-j}|b)$ are disjoint, there must be some $j$ such that
\[\left(1-\gamma^{(j)}\right)^{n-1}\geq\frac{\min_{i\in[n]}\Pr(\hat b_i> b_i\wedge\hat b_{-i}=b_{-i}|b)}{\frac1n\Pr(\hat b\not\leq b|b)}\enspace.\]
\[\frac{\Pr(\hat b_j> b_j\wedge\hat b_{-j}=b_{-j}|b)}{\Pr(\hat b_j> b_j\wedge\hat b_{-j}\leq b_{-j}|b)}\geq\frac{\min_{i\in[n]}\Pr(\hat b_i> b_i\wedge\hat b_{-i}=b_{-i}|b)}{\frac1n\Pr(\hat b\not\leq b|b)}\enspace.\]
Thus, it must be that
\[\max_{j,M\subset[n]\setminus\{j\},i\not\in M\cup\{j\}}\frac{\nu^\mu(M\cup\{i\},j,b)}{\nu^\mu(M,j,b)}\geq\frac{1-\gamma^{(>)}}{\gamma^{(>)}}\]
where $\gamma^{(>)}$ satisfies
\[\gamma^{(>)}=1-\left(\frac{\min_{i\in[n]}\Pr(\hat b_i> b_i\wedge\hat b_{-i}=b_{-i}|b)}{\frac1n\Pr(\hat b\not\leq b|b)}\right)^{\frac1{n-1}}\enspace.\]
Consequently,
\[\sup_{\hat b} \left|\frac{\rhob(\hat b)}{b_i}\right|\geq\frac{1-\gamma^{(>)}}{\gamma^{(>)}}\]
as desired.
\end{proof}

We now have the tools to prove that a resampling distribution that optimizes payments subject to a precision bound also optimizes them subject to a welfare approximation or revenue approximation bound:
\begin{proof}[of Lemma~\ref{lem:a-sp-swrapx}] For clarity, we argue under the nice distribution assumption. Subject to $\min_{i\in[n]}\Pr_\mu(\hat b_i> b_i\wedge\hat b_{-i}=b_{-i}|b)\geq\alpha>2^{-n}$, the BKS transformation achieves
\[\sup_{\hat b} \left|\frac{\rhob[\fbks](\hat b)}{b_i}\right|=\frac{\alpha^{\frac1n}}{1-\alpha^{\frac1n}}\enspace,\]
so any optimal distribution must do at least as well.

Let $\mu$ be some resampling distribution. If $\Pr_\mu(\hat b\not\leq b|b)\neq 0$, either
$$\frac{\Pr_\mu(\hat b=b|b)}{\Pr_\mu(\hat b\leq\bar b|\bar b)}>\alpha\enspace,$$
or
$$\frac{\min_{i\in[n]}\Pr_\mu(\hat b_i> b_i\wedge\hat b_{-i}=b_{-i}|b)}{\Pr_\mu(\hat b\not\leq b|b)}\geq \alpha\enspace.$$

In the first case, applying Lemma~\ref{lem:babaioff-relwrgamma} gives
$$\sup_{\hat b} \left|\frac{\rhob(\hat b)}{b_i}\right|\geq\frac{1-\gamma^{(=)}}{\gamma^{(=)}}>\frac{\alpha^{\frac1n}}{1-\alpha^{\frac1n}}=\sup_{\hat b} \left|\frac{\rhob[\fbks](\hat b)}{b_i}\right|$$
and therefore $\mu$ cannot be optimal.

In the second case, Lemma~\ref{lem:babaioff-relwrgamma} and the assumption that $\alpha>2^{-n}\geq\frac1{n^n}$ gives
\begin{eqnarray*}
 \gamma^{(>)}&\leq&1-(n\alpha)^{\frac1{n-1}}\\
&<&1-(\alpha^{-\frac1n}\alpha)^{\frac1{n-1}}\\
&=&1-\alpha^{\frac1{n}}\enspace.
\end{eqnarray*}
Thus, $\gamma^{(>)}<1-\alpha^{\frac1{n}}$, so
$$\sup_{\hat b} \left|\frac{\rhob(\hat b)}{b_i}\right|\geq\frac{1-\gamma^{(>)}}{\gamma^{(>)}}>\frac{\alpha^{\frac1n}}{1-\alpha^{\frac1n}}=\sup_{\hat b} \left|\frac{\rhob[\fbks](\hat b)}{b_i}\right|$$
so again $\mu$ cannot be optimal.

It follows that any optimal distribution $\mu$ must have $\Pr(\hat b\not\leq b|b)=0$ and, therefore
$$\min_{i\in[n]}\macroCPr{\hat b_i> b_i\wedge\hat b_{-i}=b_{-i}}{b}=\macroCPr{\hat b=b}{b}\enspace.$$
Thus, a distribution which wishes to optimize the worst-case normalized payment subject to $\Pr(\hat b=b|b)\geq\alpha$ will also optimize payments subject to $\min_{i\in[n]}\Pr(\hat b_i> b_i\wedge\hat b_{-i}=b_{-i}|b)\geq\alpha$, and will have $\Pr(\hat b=b|b)=\min_{i\in[n]}\Pr(\hat b_i> b_i\wedge\hat b_{-i}=b_{-i}|b)$.
\end{proof}

\section{Analysis Definitions, Facts, and Lemmas}
\label{sec:a-analysis}

This section provides a limited background on analysis concepts.

\subsection{Measures and Integrals}
We begin with various possible set of axioms a collection of sets may satisfy,
and their technical names. 

\begin{definition}[$\sigma$-algebra]
The $\sigma$-algebra over a set $U$ is a non-empty collection $\Sigma$ of
subsets of $U$ that is closed under complementation and countable union of its
members. The pair $(U,\Sigma)$ is called a measurable space. 
\end{definition}
\begin{definition}[Generated $\sigma$-algebra]
Given a set $U$ and a collection of subsets $F$ of $U$, there is a unique
smallest $\sigma$-algebra over $U$ containing all the elements of $F$. This
$\sigma$-algebra is denoted by $\sigma(F)$ and is called as the $\sigma$-algebra
generated by $F$.
\end{definition}
\begin{definition}[Borel $\sigma$-algebra]
The Borel $\sigma$-algebra $\mathcal{B}(U)$ of a metric space 
$U$ is the $\sigma$-algebra generated by the collection of all open sets of $U$.
\end{definition}

\begin{definition}[Measurable sets]
Once we fix a measurable space $(U,\Sigma)$, the sets $X \in \Sigma$
are called measurable sets.
\end{definition}

\begin{definition}[Measurable functions]
Given two measurable spaces $(U,\Sigma)$ and $(U',\Sigma')$, a function
$f:U\rightarrow U'$ is measurable if for each $X' \in \Sigma'$, $f^{-1}(X') \in \Sigma$. 
\end{definition}
We are now ready for the definition of a measure. 
\begin{definition}[Measure] Given a measurable space $(U,\Sigma)$, we equip it
with a {\em measure} $\nu$, which is function $\nu:\Sigma \rightarrow [0,\infty]$ that
satisfies 
\begin{enumerate}
\item $\nu(\emptyset)=0$ 
\item Countable additivity, i.e. for all countable
sequences $\{X_i\}_{i\in Z}$ of pairwise-disjoint sets in $\Sigma$, $\nu(\cup_{i
\in Z} X_i)=\sum_{i\in Z}\nu(X_i)$.
\end{enumerate}
A measure $\nu$ is said to be {\em finite} if $\nu(U)$ is finite.
\end{definition}
\begin{definition}[Probability measure]
A measure is a {\em probability measure} if $\nu(U)=1$.
\end{definition}
\begin{definition}[Signed measure]
A {\em signed measure} is a function $\nu:\Sigma \rightarrow
[-\infty,\infty]$ that satisfies $\nu(\emptyset)=0$ and countable additivity.
\end{definition}
\begin{fact}\label{fact:sp-char-smadd}
If $\nu_1$ and $\nu_2$ are finite (signed) measures, then $\nu_3(X)=\nu_1(X)-\nu_2(X)$ is a finite signed measure.
\end{fact}
\paragraph{Convention}
According to standard convention, a measure is not signed unless explicitly stated.
For the purposes of this paper, the set $U$ will always be $\Re^n$. 

Apart from the set collections defined via $\sigma$-algebras, we also need some
weaker set collections, which we define below. 
\begin{definition}[$\pi$-system]
The $\pi$-system over a set $U$ is a non-empty collection $P$ of
subsets of $U$ that is closed under finite intersection of its members, i.e.,
$X_1 \cap X_2 \in P$ whenever $X_1$ and $X_2 \in P$. 
\end{definition}
\begin{definition}[$\lambda$-system, or Dynkin system]
The $\lambda$-system over a set $U$ is a non-empty collection $L$ of
subsets of $U$ that is closed under complementation and countable disjoint union of its
members.
\end{definition}
\begin{fact}[Dynkin's theorem]
If $P$ is a $\pi$-system and $L$ is a $\lambda$-system over the same set $U$, 
and $P \subseteq L$, then $\sigma(P) \subseteq L$, i.e., the $\sigma$-algebra generated by $P$ is
contained in $L$. 
\end{fact}

The Hahn and Jordan decompositions decompose a signed measure into two measures. They will be useful when we discuss the integral with respect to a signed measure.
\begin{fact}[Hahn decomposition theorem]\label{fact:ana-hahn} 
The {\em Hahn decomposition} of a signed measure $\nu$ over a measurable space
$(U,\Sigma)$ consists of two sets $P, N \in \Sigma$ such that
$P\cup N=U$, $P\cap N=\emptyset$, and for all measurable sets $X \subseteq P$,
$\nu(X)\geq0$ and for all measurable sets $X \subseteq N$, $\nu(X)\leq0$. 
The Hahn decomposition is guaranteed to exist
and be unique (up to a set of measure 0) 
\end{fact}
\begin{fact}[Jordan decomposition theorem]\label{fact:ana-jordan} This
theorem is a consequence of Hahn decomposition theorem, and states that 
every signed measure $\nu$ can be decomposed as two (non-negative) measures
$\nu^+(X)=\nu(X\cap P)$ and $\nu^-(X)=-\nu(X\cap N)$, where $P$ and $N$ are the
Hahn decomposition of $\nu$. The measures satisfy $\nu(X)=\nu^+(X)-\nu^-(X)$.
The Jordan decomposition is guaranteed to exist and to be unique, and at least
one of $\nu^+$ and $\nu^-$ is guaranteed to be a finite measure. If $\nu$ is
finite, then both $\nu^+$ and $\nu^-$ are finite.
\end{fact}

\begin{definition}[Characteristic Function]\label{def:ana-char-func}
The {\em characteristic function} $1_S(x)$ of a set $S$ is the function that 
is 1 if $x\in S$ and zero elsewhere, i.e.
$$1_S(x)=\begin{cases}
1,&x\in S\\
0,&otherwise.
\end{cases}$$
\end{definition}

\begin{definition}[Simple Function]\label{def:ana-simp-func}
Given a measurable space $(U,\Sigma)$, a function $s:U\rightarrow \Re$ is a {\em simple function} 
if it can be written as a finite linear combination of indicator function of
measurable sets. That is, $$s(x)=\sum_{k=1}^n a_k1_{S_k}(x)$$
for finite sequences of measurable sets $\{S_k\} \in \Sigma$ and coefficients $\{a_k\} \in
R$.
\end{definition}

\begin{fact}\label{fact:f-lim-simp}
For any non-negative, measurable function $f$, there is a monotonic increasing sequence of 
non-negative simple functions $\{s_k\}$ such that
$$f=\lim_{k\rightarrow\infty}s_k\enspace.$$
\end{fact}

\begin{definition}[Integral]\label{def:integral} Given a measurable space $(U,\Sigma)$, the 
integral of a function $f:U\rightarrow \Re$ with respect to a measure $\nu$ 
is defined incrementally. For any measurable set $X$, the integral of $1_X$ is
$$\int_U 1_Xd\nu=\nu(X)\enspace.$$
For any simple function $s:U\rightarrow \Re$,
$$\int_U sd\nu=\sum_{k=1}^na_k\nu(X_k)\enspace.$$
For a general non-negative function $f:U\rightarrow \Re$,
$$\int_U fd\nu=\sup\left\{\int_U sd\nu\quad:\quad0\leq s\leq f\mbox{ and }s\mbox{ is simple}\right\}\enspace.$$
For general $f$, let $f^+(x)=\max(f(x),0)$ and $f^-(x)=\max(-f(x),0)$, i.e. $f^+$ and $f^-$ are the positve and negative parts of $f$ respectively. Then
$$\int_U fd\nu=\int_U f^+d\nu-\int_U f^-d\nu\enspace.$$
Finally, for some measurable set $Y$, 
$$\int_Y fd\nu=\int_U fd\nu_Y$$
where $\nu_Y(X)=\nu(U\cap Y)$.
\end{definition}

\begin{fact}[Monotone Convergence Theorem]
For any countable, monotone sequence of measurable functions $\{f_k\}$ (that is, sequence where $f_k\geq f_{k-1}$ pointwise),
\[\lim_{k\rightarrow\infty}\int f_kd\nu=\int\lim_{k\rightarrow\infty}f_kd\nu\enspace.\]
\end{fact}
The following fact follows because $g_k=\sum_{i=1}^kf_i$ satisfies the monotone convergence theorem:
\begin{fact}\label{fct:ana-integral-countadd}
For any countable sequence of nonnegative measurable functions $\{f_k\}$
\[\sum_{k=1}^\infty\int f_kd\nu=\int\sum_{k=1}^\infty f_kd\nu\enspace.\]
\end{fact}

\begin{fact}\label{fact:ana-int-cu} Let $\{X_k\}$ be a countable sequence of disjoint sets. Then
$$\sum_k\int_{X_k} fd\nu=\int_{\cup_kX_k}fd\nu\enspace.$$
\end{fact}

\begin{definition}[Integral with respect to a Signed Measure] The integral of a function $f$ with respect to a signed measure $\nu$ is
$$\int_U fd\nu=\int_U fd\nu^+-\int_U fd\nu^-\enspace,$$
where $\nu^+$ and $\nu^-$ are the Jordan decomposition of $\nu$.
\end{definition}

\subsubsection{Densities and Derivatives}

\begin{definition}[Absolute continuity]
Given a signed measure $\nu$ and a measure $\mu$ on the same measurable space, 
$\nu$ is absolutely continuous w.r.t. $\mu$, if for every measurable set $V$ where $\mu(V)
= 0$, we have $\nu(V) = 0$.
\end{definition}
We now state below the Radon-Nikodym theorem the way we use it, though the
theorem itself is more general. 
\begin{fact}[Radon-Nikodym Theorem]
The Radon–Nikodym theorem states that given a finite signed measure $\nu$ and a
finite measure
$\mu$ on the same measurable space such that $\nu$
is absolutely continuous w.r.t. $\mu$, the measure $\nu$ has a density, or
``Radon-Nikodym derivative", with respect to $\mu$, i.e., there exists a
$\mu$-measurable function $\rho$ taking values in $[0, \infty]$, such
that for any $\mu$-measurable set $X$ we have
$$\nu(X)=\int_X\rho d\mu\enspace.$$
\end{fact}
\begin{fact}
If $\rho$ is a Radon-Nikodym derivative of measure $\nu$ w.r.t. measure $\mu$, then
$$\int_X f(x)d\nu=\int_X \rho(x)f(x)d\mu$$
wherever $\int_X f(x)d\nu$ is well defined.
\end{fact}
\subsection{Almost Everywhere}

\begin{definition}[Almost Everywhere]
A property $P(s)$ is said to hold {\em almost everywhere} on a set $S$ if the
subset of $S$ on which $P(s)$ is false has measure zero (or is contained in a
set that has measure 0). It is abbreviated {\em a.e.}.The exact measure used
will become clear from the context. 
\end{definition}
\begin{definition}[Almost Surely]
If a property $P(s)$ is false with probability 0 with respect to some distribution over $s$, then it is said to hold {\em almost surely}. This is equivalent to saying $P(s)$ is true almost everywhere with respect to the probability measure associated with the distribution.
\end{definition}

\begin{fact}\label{fact:ana-nnpfzeroae} For a non-negative measurable function
$f$ and measure $\mu$, $\int fd\mu=0$ if and only if $f(x)=0$ almost everywhere.
\end{fact}
\begin{fact}\label{fact:ana-fzeroae}
For any measurable function $f$ and signed measure $\nu$, if $\int_Xfd\nu=0$ for all measurable $X$, then $f=0$ almost everywhere.
\end{fact}
The second fact follows from the first by a standard argument --- decompose $f$ into its positive and negative parts and decompose $\nu$ according to its Hahn decomposition. This partitions the space into four sets over which the integral may be written as a non-negative function with respect to a non-negative measure. Apply Fact~\ref{fact:ana-nnpfzeroae} to each of the four sets.

\subsection{Extrema}

\begin{definition}[Supremum/Infimum]
For a set $S$, the {\em supremum} of $S$, denoted $\sup S$, is the smallest
value $x$ such that $x\geq s$ for all $s\in S$. Similarly, the {\em infimum} of
$s$ is the largest value $x$ such that $x\leq s$ for all $s\in S$.
\end{definition}
\begin{definition}[Limit Superior/Inferior]
For a real-valued function $f:\Re^n\rightarrow\Re$, the {\em limit superior}, denoted $\limsup_{u\rightarrow b}f(u)$, may be defined as follows:
$$\limsup_{u\rightarrow b}f(u)=\lim_{\epsilon\rightarrow0}\left(\sup_{u\in BALL(b,\epsilon)}f(u)\right)$$
where $BALL(b,\epsilon)$ is the open ball of radius $\epsilon$ centered at $b$.
It is an upper bound on the limit of $f(u_i)$ for any sequence of values
$\{u_i\}$ that converges to $b$. The $\liminf$ is defined similarly. Note that
while the limit may not exist as $u\rightarrow b$, the $\limsup$ and $\liminf$
are always well defined for real-valued functions.
\end{definition}

It is natural to generalize $\sup$ and $\limsup$ to almost everywhere:
\begin{definition}[Essential Supremum/Infimum]
The {\em essential supremum} of a set $S$, denoted $\esssup S$, is the smallest value $x$ such that $x\geq s$ almost everywhere, i.e. the set of values $T=\{s|s\in S\mbox{ and }s>x\}$ has measure zero. The {\em essential infimum} $\essinf$ is defined similarly.
\end{definition}
\begin{definition}[limesssup/limessinf]
For a function $f:\Re^n\rightarrow\Re$, the $\limesssup_{u\rightarrow b}f(u)$ can be defined as follows:
$$\limesssup_{u\rightarrow b}f(u)=\lim_{\epsilon\rightarrow0}\left(\esssup_{u\in BALL(b,\epsilon)}f(u)\right)\enspace.$$
It can be understood as a version of the $\limsup$ that will ignore values that $f(x)$ only attains on sets with measure zero. The $\limessinf$ is defined similarly. Like the $\limsup$ and $\liminf$, the $\limesssup$ and $\limessinf$ are always well defined for real-valued functions.
\end{definition}

\end{document}